\documentclass{article}

\usepackage[english]{babel}

\usepackage[letterpaper,top=2cm,bottom=2cm,left=2cm,right=2cm,marginparwidth=1.75cm]{geometry}

\usepackage{graphicx}
\usepackage[colorlinks=true, allcolors=blue]{hyperref}

\usepackage{amsfonts}
\usepackage{amsmath}
\usepackage{amsfonts}
\usepackage{amsthm}

 \usepackage{subcaption}
\usepackage{url}
\usepackage{authblk}
\usepackage{bm}

\usepackage{graphicx}%
\usepackage{multirow}%
\usepackage{amsmath,amssymb,amsfonts}%
\usepackage{amsthm}%
\usepackage{mathrsfs}%
\usepackage[title]{appendix}%
\usepackage{xcolor}%
\usepackage{textcomp}%
\usepackage{manyfoot}%
\usepackage{booktabs}%
\usepackage{algorithm}%
\usepackage{algorithmicx}%
\usepackage{algpseudocode}%
\usepackage{listings}%
\usepackage{url}
\usepackage{float}
\usepackage{authblk}
\usepackage{hyperref}
\usepackage{adjustbox}
\usepackage{bbm}
\usepackage{todonotes}
\usepackage{upgreek}
\usepackage{wrapfig}
\usepackage{cite}

\newtheorem{theorem}{Theorem}

\newtheorem{lemma}{Lemma}
\newtheorem{proposition}{Proposition}

\newtheorem{corollary}{Corollary}

\theoremstyle{definition}

\theoremstyle{remark}

\algrenewcommand\algorithmicrequire{\textbf{Input:}}
\algrenewcommand\algorithmicensure{\textbf{Output:}}

\title{2-Adic Tree Embeddings: A Modeling and Optimization Framework for Computational Phylogenetics}

\author[1,$\dagger$]{P. Skums}

\affil[1]{School of Computing, University of Connecticut, Storrs, CT, USA}
\affil[$\dagger$]{Email: pavel.skums@uconn.edu}

\date{}

\begin{document}
\maketitle

\begin{abstract}
Phylogenetic inference is a fundamental problem in computational biology, with applications ranging from epidemiology and oncology to developmental biology and comparative genomics. Many major phylogenetic inference problems are NP-hard; consequently, existing methods often explore tree space using local search or MCMC sampling based on a relatively small repertoire of combinatorial rearrangements. The size and ruggedness of tree space make it difficult for such general-purpose heuristics to reliably reach optimal or near-optimal solutions, creating a major bottleneck in computational phylogenetics. Recently, there has been growing interest in alternative approaches that embed trees into numerical spaces and apply numerical optimization techniques. Although promising, such methods are not yet broadly practical and remain limited in their applicability.

We propose a new framework for phylogenetic inference based on embeddings of trees into the non-Archimedean space of $2$-adic numbers. The hierarchical organization of $2$-adic space naturally encodes tree topologies and provides algebraic representations of tree metrics and transformations, while its rich topological, algebraic, and measure-theoretic structure provides a natural foundation for a coherent theory of tree calculus and optimization.

This paper is the first in a series intended to develop the theory and applications of $2$-adic tree embeddings. Here, we introduce the concept, establish several of its topological, probabilistic, and algebraic properties, and show that $2$-adic embeddings support natural representations of tree distributions, algebraic operations, and broad classes of tree rearrangements that extend standard operations such as NNI, SPR, and TBR. 

Building on these results, we introduce a general local search framework in which tree neighborhoods are generated by parametric families of $2$-adic polynomials. We show that even infinitely many transformations induce only quadratically many distinct neighboring trees, which can be generated in polynomial time. These results provide a mathematically rigorous and flexible framework for phylogenetic inference and suggest a new approach to representing and exploring tree space.
\end{abstract}

\section{Introduction}

Tree optimization, i.e. search for a tree that satisfies specified constraints while optimizing a given objective function, arises in many fields \cite{taha2013operations,bai-etal-2021-syntax,mcdonald2005non,tai-etal-2015-improved}. A major application domain is evolutionary biology, including molecular epidemiology \cite{stadler2013birth}, cancer genomics \cite{lahnemann2020eleven}, developmental biology \cite{stadler2021phylodynamics}, and evolutionary genomics \cite{ingman2000mitochondrial}. In these areas, tree optimization arises primarily in {\it phylogenetic inference}, where evolutionary relationships among a given set of taxa are represented by a tree and the goal is to infer the topology and associated parameters that best explain the observed data under a specified evolutionary model \cite{Felsenstein:Book:2003,semple2003phylogenetics,drummond2015bayesian}. The corresponding objectives range from relatively simple criteria, such as discrepancies between observed distances and tree-induced distances \cite{fitch1967construction,desper2004theoretical}, to likelihood- and posterior-based objectives derived from complex evolutionary, demographic, and epidemiological models \cite{volz2013viral,grenfell2004unifying}.

Many important optimization problems in phylogenetic inference are NP-hard \cite{day1987computational,day1986computational,chor2005maximum,sashittal2023startle,el2018inferring}. Consequently, practical algorithms often explore tree space using local search or Markov chain Monte Carlo (MCMC) sampling. Starting from a current topology, they generate candidate trees through fixed local rearrangements, including nearest-neighbor interchange (NNI), subtree pruning and regrafting (SPR), and tree bisection and reconnection (TBR). Such moves are central to widely used general-purpose phylogenetic software \cite{stamatakis2014raxml,minh2020iq,Guindon:SB:2010,pineiro2020very,drummond2007beast,ronquist2012mrbayes}, as well as specialized tools for particular application domains \cite{jahn2016tree,liu2024nestedbd}. These approaches have proved effective across many applications and remain a computational foundation of modern phylogenetic inference. Over the years, existing methods have been refined extensively to improve computational efficiency via effective usage of multi-core and distributed architectures, vectorization, and GPU acceleration, together with highly optimized algorithms for objective evaluation and parameter optimization.

Yet, efficient exploration of tree space remains a major bottleneck. The principal challenges are the superexponential size and ruggedness of the solution space, which make it difficult for general-purpose heuristics to reliably reach optimal or near-optimal solutions \cite{spirin2024phylobench,zhou2018evaluating,liu2024influence,fabreti2022convergence}.  The progress in this domain has been incremental, as most methods still rely on local search based on a relatively small repertoire of standard tree rearrangements. Further major advances may therefore require a more fundamental change in how tree spaces are represented and explored.

There are, however, alternative ways to formulate the problem. One approach is to encode trees numerically, thereby recasting phylogenetic inference as an optimization problem over numerical variables. A prominent example is Integer Linear Programming (ILP) \cite{gusfield2019integer}. ILP provides a flexible framework for expressing complex combinatorial constraints, can certify global optimality, and benefits from highly optimized commercial and open-source solvers. Its main limitation is scalability: many formulations require large numbers of auxiliary variables and constraints merely to ensure that feasible numerical solutions correspond to valid trees.

Recently, there has been growing interest in alternative compact and computationally practical numerical embeddings of phylogenetic trees \cite{chauve2025vector,penn2025phylo2vec,penn2023leaping,liu2022analyzing} that can support gradient-based or hill-climbing algorithms for exploring tree space more efficiently \cite{penn2023leaping,penn2025phylo2vec}. Most existing approaches use vector or matrix representations embedded in, or relaxed to, Euclidean spaces. Complementary work has considered hyperbolic embeddings, whose geometry is particularly well suited to representing hierarchical structures \cite{matsumoto2021novel,macaulay2023fidelity,mimori2023geophy}.

These approaches are promising, but the field remains largely at the proof-of-concept stage. A major challenge is that current embeddings often fail to induce sufficiently simple or well-behaved objective functions and constraints \cite{macaulay2023fidelity,penn2025phylo2vec}, limiting the use of established mathematical and computational optimization techniques. As a result, some methods remain conceptually close to traditional tree-search heuristics and inherit many of their limitations \cite{penn2025phylo2vec}, whereas others are computationally expensive or applicable only in restricted settings \cite{penn2023leaping}.


In this paper, we introduce a new approach to tree representation and optimization based on embeddings into the non-Archimedean space of {\it $p$-adic numbers}. For a prime $p$, the field $\mathbb{Q}_p$ of $p$-adic numbers is the completion of the rational numbers with respect to the $p$-adic norm $|x|_p$, which measures, roughly speaking, the divisibility of $x$ by $p$ \cite{koblitz2012p,Schikhof_1985,gouvea1997p,robert2013course}. The $p$-adic numbers play a central role in algebra and number theory and have also found applications in areas ranging from theoretical physics \cite{brekke1993p,dragovich2009p} to theoretical biology \cite{dragovich2009p,torresblanca20265} and cognitive science \cite{khrennikov2004information}. Several previous studies have explored $p$-adic methods in genomics and phylogenetics, particularly $p$-adic representations of the genetic code \cite{dragovich2009p,torresblanca20265,khrennikov2007genetic,dragovich2019p}. However, to the best of our knowledge, $p$-adic analysis has not been systematically developed as a practical computational framework for tree optimization and related problems in bioinformatics and discrete optimization. 

Topologically, the space $\mathbb{Q}_p$ has an intrinsic hierarchical structure induced by its nested ultrametric balls, naturally supporting tree embeddings. 
More importantly, its combined topological, algebraic, and measure-theoretic structure provides a natural foundation for a coherent theory of tree calculus and optimization. 

In this paper, we focus primarily on binary trees, the standard setting in phylogenetics, and therefore concentrate on $2$-adic numbers. The underlying methodology, however, is expected to extend naturally to other $p$-adic spaces and to broader classes of tree optimization problems.

The paper is organized into several parts that develop the topological, measure-theoretic, and algebraic properties of $2$-adic tree embeddings and present a general methodology for efficient tree inference based on $2$-adic local search. The paper is structured as follows:

\begin{itemize}

\item[1)] We introduce {\it $2$-adic tree embeddings} of rooted binary trees and establish their basic properties, including the equivalence of alternative definitions and the existence and uniqueness of embeddings satisfying prescribed conditions.

\item[2)] We study the {\it measure-theoretic properties} of $2$-adic tree embeddings. In particular, for a given tree, we analyze the subspace of its $2$-adic embeddings and estimate its Haar measure, which turns out to be closely related to the Sackin index, a widely used measure of tree balance. We further show that every Aldous $\beta$-splitting distribution with $\beta>-1$, including the Yule distribution and other biologically relevant models, can be realized by a suitable probability distribution on the $2$-adic space: independent sampling from this distribution yields an embedding of a random tree with the corresponding $\beta$-splitting distribution.

\item[3)] We investigate {\it algebraic operations} on tree embeddings, focusing on the asymptotic behavior of large sums and Hadamard products. These operations exhibit unexpectedly different limiting behavior, with the two limits concentrated at opposite ends of the tree-balance spectrum. The sum of $n$ tree embeddings drawn independently from any non-degenerate $2$-adic distribution converges, as $n\to\infty$, to the Haar tree distribution, which favors balanced trees, thus providing a natural tree analogue of the Central Limit Theorem. In contrast, the Hadamard product of $n$ independent, identically Haar-distributed tree embeddings converges to the uniform distribution on caterpillar trees, which are maximally unbalanced.

\item[4)] We introduce the notions of {\it polynomial and locally affine tree  transformations}. We show that standard tree rearrangements (e.g., NNI and SPR) can be realized by such algebraic transformations. More generally, we prove that any operation generalizing NNI, SPR, and TBR by rearranging subtrees while preserving their internal structure admits such realization.

\item[5)] Building on these results, we propose a general scheme for {\it local search in tree space}, in which the neighborhood of a given tree is defined by its images under $2$-adic polynomials from a specified parametric family. This framework generalizes standard phylogenetic local search and provides greater flexibility in adapting the neighborhood structure to particular inference problems. For a broad class of such transformation families, we prove that, even when the family itself is infinite, the neighborhood of any $n$-leaf tree contains $O(n^2)$ distinct trees and can be generated in polynomial time. We also provide an explicit algorithm for generating this neighborhood. 
\end{itemize}

Together, these results lay the theoretical foundations for a general and mathematically rigorous framework for phylogenetic inference, applicable to a broad range of problems.
 

 \vspace{-2ex}
\section{The concept of $2$-adic embeddings of binary trees}
\vspace{-1ex}

\subsection{Basic definitions}\label{sec:def}
\vspace{-1ex}

Let $\mathbb{N}$, $\mathbb{Z}$ and $\mathbb{Q}$ be the sets of natural, integer and rational numbers, respectively. Every number $x \in \mathbb{Q}$ can be uniquely expressed as

\vspace{-1ex}
\begin{equation}
x = 2^{k}\frac{a}{b}    
\end{equation}

\noindent
where $k \in \mathbb{Z}$ and $a,b \in \mathbb{N}$ are odd. The integer $k$ is called the {\it $2$-adic valuation} of $x$ and is denoted by $v_2(x)$.
The {\it $2$-adic norm} of $x$ is defined as $|x|_2 = 2^{-v_2(x)}$; intuitively, it measures how divisible $x$ is by $2$.
This norm satisfies the {\it ultrametric inequality}, a stronger form of the usual triangle inequality:

\vspace{-2ex}
\begin{equation}\label{eq:triangnorm}
    |x+y|_2 \leq \max \{|x|_2,|y|_2\} \;\;\; \text{ for any } x,y \in \mathbb{Q}.
\end{equation}

\noindent
This inequality can be rewritten in terms of valuations as

\begin{equation}\label{eq:triangval}
    v_2(x+y) \geq \min \{v_2(x),v_2(y)\} \;\;\; \text{ for any } x,y \in \mathbb{Q}.
\end{equation}

\noindent
Moreover, if $|x|_2 \ne |y|_2$, then inequalities (\ref{eq:triangnorm}) and (\ref{eq:triangval}) become equalities.

The set of {\it $2$-adic numbers} $\mathbb{Q}_2$ is defined as the completion of rational numbers with respect to the $2$-adic norm $|\cdot|_2$. Every $2$-adic number can be uniquely expressed as an infinite series, called its {\it $2$-adic expansion}:

\vspace{-1ex}
\begin{equation}\label{eq:padicexpans}
    x = \sum_{i=-k}^{\infty} a_i2^i: a_i \in \{0,1\}, k\in \mathbb{N}.
\end{equation}

\noindent
Unlike binary expansions of real numbers, the $2$-adic expansion \eqref{eq:padicexpans} contains only finitely many digits after the radix point and infinitely many digits before it.


We will primarily use the subspace $\mathbb{Z}_2 \subseteq \mathbb{Q}_2$ of {\it $2$-adic integers}, defined as the set of $2$-adic numbers whose expansions involve only non-negative powers of $2$:

\vspace{-1ex}
\begin{equation}
    \mathbb{Z}_2 = \left\{z = \sum_{i=0}^{\infty} a_i2^i: a_i \in \{0,1\}\right\} = \{z\in \mathbb{Q}_2: |z|_2 \leq 1\}.
\end{equation}


$\mathbb{Z}_2$ contains all ordinary integers but is strictly larger ($\mathbb{Z} \subset \mathbb{Z}_2$): for example, numbers such as $\tfrac{1}{3}$ and $\sqrt{17}$ are 2-adic integers.  More generally, a rational number $\tfrac{a}{b} \in \mathbb{Q}$ is a $2$-adic integer whenever the denominator $b$ is odd. In general, the set of invertible 2-adic integers, denotes as $\mathbb{Z}_2^{\times}$, strictly contains the set of odd integers

A {\it ball} $B_r(x)$ with center $x$ and radius $r$ is defined in the usual way as $B_r(x) = \{y\in \mathbb{Z}_2 : |x-y|_2 \leq r\}$. Balls in an ultrametric space are hierarchically organized and possess properties that distinguish them from those in standard metric spaces. In particular, every point inside a ball is its center (i.e. $B_r(x) = B_r(y)$ for any $y \in B_r(x)$), and the set of 2-adic balls $\mathcal{B}(\mathbb{Z}_2)$ form a {\it laminar family}, i.e. any two balls are either disjoint or one is a subset of the other. 

Furthermore, balls in $\mathbb{Z}_2$ coincide with residue classes modulo $2^k$; that is, each ball can be represented in the canonical form

\vspace{-1ex}
\begin{equation}
    B_{2^{-k}}(x) = a + 2^k\mathbb{Z}_2 = \{a + 2^k y: y\in \mathbb{Z}_2 \},
\end{equation}

\noindent 
where $k\geq 0$ and $a\in \{0,1,...,2^k-1\}$ is the integer such that $|x-a|_2 \leq 2^{-k}$ (Fig. \ref{fig:pspace}).

\begin{wrapfigure}{r}{0.5\textwidth}
  \begin{center}
    \includegraphics[width=0.48\textwidth]{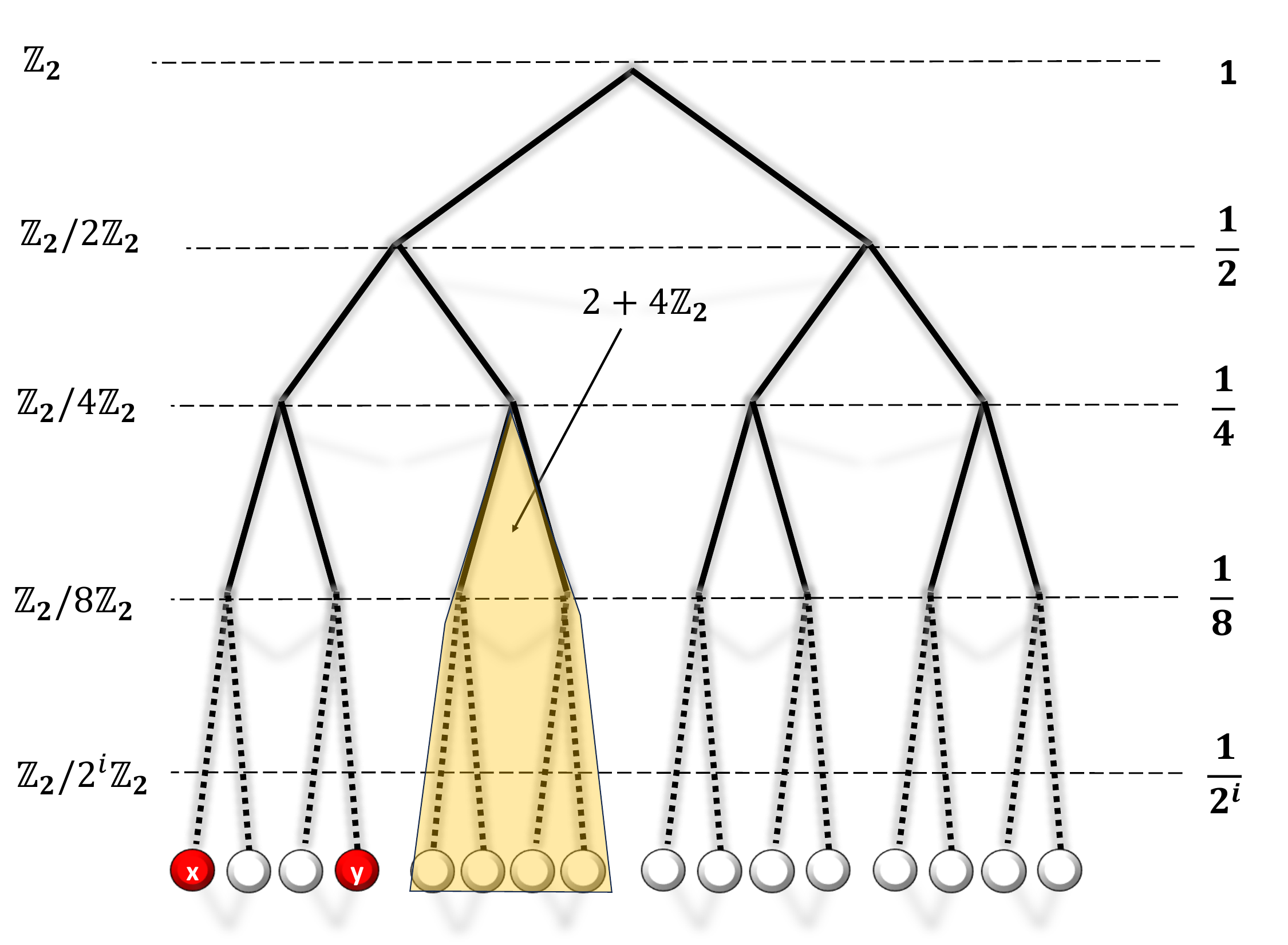}
  \end{center}
  \caption{\small
The space $\mathbb{Z}_2$ represented as an infinite binary tree $\mathcal{Z}$. 
For the two highlighted leaves $x, y \in \mathbb{Z}_2$, we have $|x - y|_2 = 1/4$. 
The ball $B = 2 + 4\mathbb{Z}_2$ is highlighted in orange. Dashed horizontal lines indicate successive $2$-adic scales.} \label{fig:pspace}
\end{wrapfigure}

For $X\subseteq \mathbb{Z}_2$, $B(X)$ is the minimal ball containing $X$. In particular, we denote by $B(x,y)$ the minimal ball containing both $x$ and $y$: $B(x,y) = B_{|x-y|_2}(x) = B_{|x-y|_2}(y)$. Laminarity of 2-adic balls imply that such ball is unique. Note also that every ball $B = a+2^k\mathbb{Z}_2$ can be represented as the union of two maximal disjoint balls $B^{(0)} = a+2^{k+1}\mathbb{Z}_2$ and $B^{(1)} = a+2^k+2^{k+1}\mathbb{Z}_2$; the latter will be referred to as the {\it left and right children balls} of $B$, respectively.

A {\it branching ball of a finite set} $X\subset\mathbb{Z}_2$ is a ball whose two children balls both intersect $X$. The laminarity of 2-adic balls and the fact that every 2-adic ball has exactly two children balls imply that there are exactly $|X|-1$ branching balls. The branching ball of a vector $\bm{s} = (s_1,...,s_n)$ is defined as a branching ball of the underlying set $\{s_1,...,s_n\}$

Besides being a metric space, $\mathbb{Z}_2$ is a {\it ring}, where addition, subtraction, multiplication, and division (when possible) are defined similarly to real numbers, and the $2$-adic norm determines convergence of infinite series (for example, in $\mathbb{Z}_2$, one has $-1 = ...1111 =  \sum_{i=0}^{\infty} 2^i$). Geometrically, $\mathbb Z_2$ can be visualized as an infinite rooted binary tree $\mathcal{Z}_2$ (called {\it the Bruhat-Tits or 2-adic tree}), whose leaves correspond to $2$-adic integers and internal nodes to $2$-adic balls (Fig.~\ref{fig:pspace}). 

Finally, $\mathbb{Z}_2$ is a probability space, with a {\it Haar probability measure} $\mu$ defined on balls of $\mathbb{Z}_2$ as follows:

\begin{equation}
    \mu(a + 2^k\mathbb{Z}_2) = 2^{-k}.
\end{equation}




\vspace{-2ex}
\subsection{$2$-adic embeddings of trees and their basic properties}\label{sec:embed}
\vspace{-1ex}

Consider a rooted binary tree $T = (V(T),E(T))$ with the leaf set $L(T) = [n]$, internal nodes $I(T)$ and the root $r(T)$. For $v\in I(T)$, the set $C_v = C_{T,v}$ of leaf descendants of $v$ is {\it a clade}; the set of all clades is denoted by $\mathcal{C}(T)$.

Let $\mathbf{s}=(s_1,\ldots,s_n)\in\mathbb{Z}_2^n$ be a vector of distinct $2$-adic integers. It defines a tree $\mathcal{T}(\mathbf{s})$, which admits several equivalent descriptions:

\begin{itemize}
    \item[(a)] $\mathcal{T}(\mathbf{s})$ is obtained by {\it smoothing} the subtree of the $2$-adic tree $\mathcal{Z}_2$ spanned by $s_1,\ldots,s_n$, i.e., by suppressing all vertices of outdegree one  (Fig.~\ref{fig:embed}, left). 
    
    \item[(b)] $\mathcal{T}(\mathbf{s})$ is the {\it hierarchical clustering tree} of $s_1,\ldots,s_n$ under the $2$-adic metric. By ultrametricity, minimum, maximum, and average linkage induce the same hierarchy.

    \item[(c)] $\mathcal{T}(\mathbf{s})$ is the {\it binary radix tree} of the $2$-adic expansions of $s_1,\ldots,s_n$, viewed as binary words. 
\end{itemize}

We say that $\bm{s}$ is {\it a point embedding} of the tree $T$ if $T = \mathcal{T}(\mathbf{s})$. A similar construction was used in \cite{bradley2009mumford,kato2005non}.

\begin{figure}
\begin{center}
\includegraphics[width=0.8\textwidth] {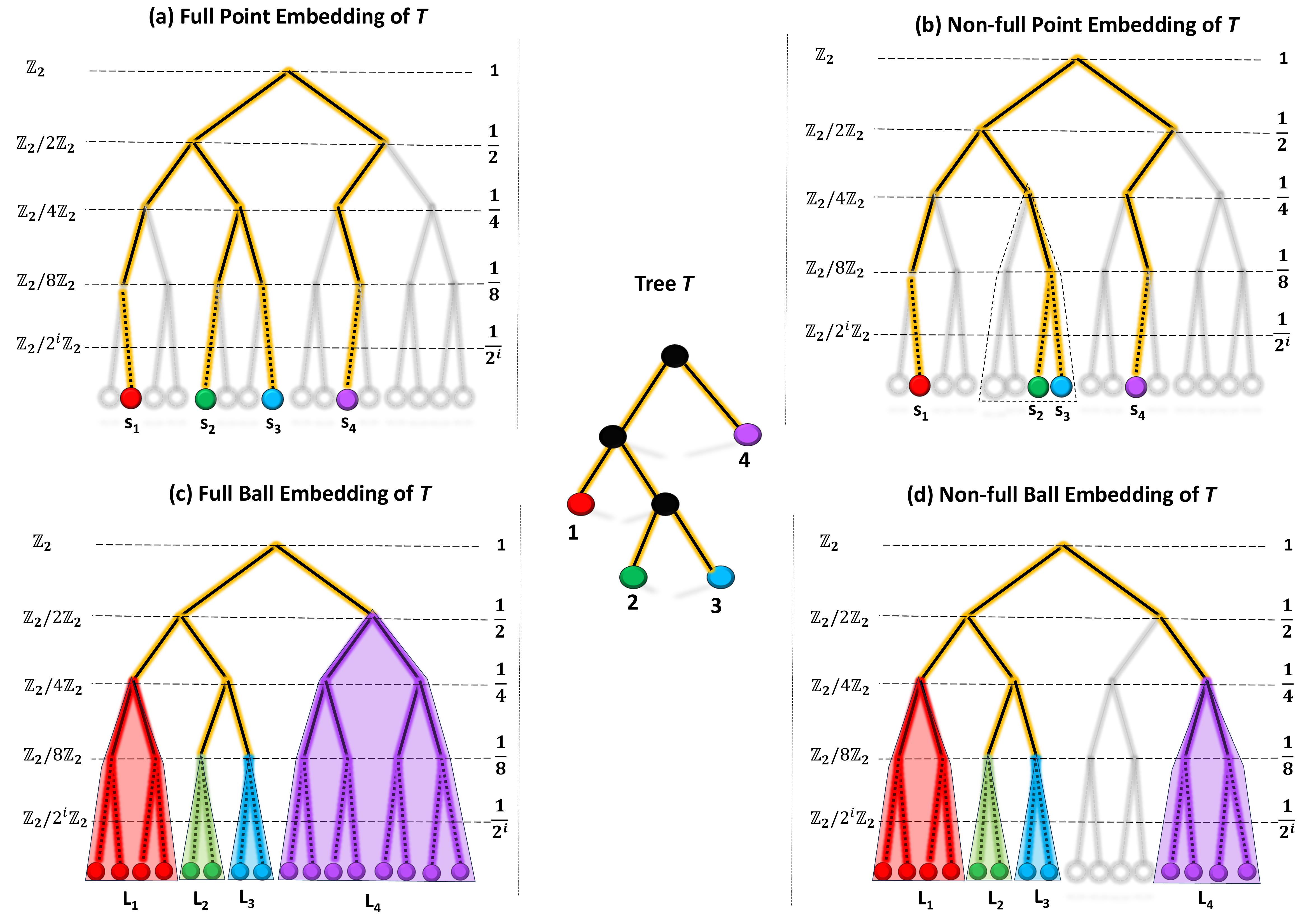}
\caption{\small $2$-adic embeddings of the same rooted binary tree $T$ ({\bf center}). {\bf (a)} A full point embedding $\bm{s}=(s_1,s_2,s_3,s_4)$, in which each leaf $i$ is represented by a coordinate $s_i\in\mathbb{Z}_2$. {\bf (b)} A non-full point embedding; the ball violating the fullness condition is indicated by a dashed line. {\bf (c)} A full ball embedding $\mathcal{L}=(L_1,L_2,L_3,L_4)$, in which each leaf $i$ is represented by a ball $L_i$ (shown in color). {\bf (d)} A non-full ball embedding, with $L_1\sqcup\cdots\sqcup L_4\neq\mathbb{Z}_2$.
\vspace{-4ex}}\label{fig:embed}
\end{center}
\end{figure}

A tree embedding can also be defined by associating its leaves with a collection of pairwise disjoint balls $\mathcal{L}={L_1,\ldots,L_n}$. Define

\begin{equation}\label{eq:clustEmbed}
    \mathcal{C}(\mathcal{L}) = \{A \subseteq [n] : B\left(\bigcup_{i\in A} L_i\right)\bigcap L_j = \emptyset \text{ for all } j\not\in A\}
\end{equation}

\noindent
We say that $\mathcal{L}$ is a {\it ball embedding} of $T$ if $\mathcal{C}(\mathcal{L})$ coincides with the clade system of $T$. Equivalently, $T=\mathcal{T}(\mathcal{L})$ is the Hasse diagram of $\mathcal{C}(\mathcal{L})$ under inclusion, with $[n]$ as the root and the singletons as leaves.
 
The ball embedding $\mathcal{L}$ is {\it full} if $\bigcup_{i=1}^nL_i=\mathbb{Z}_2$. Similarly, a point embedding $\bm{s}$ is {\it full} if every ball containing at least two coordinates of $\bm{s}$ is a branching ball. Equivalently, the subtree of $\mathcal{Z}_2$ spanned by $s_1,\ldots,s_n$ contains no vertices of outdegree one on paths between branching vertices.


We call a ball embedding $\mathcal{L}=(L_1,\ldots,L_n)$ and a point embedding $\bm{s}=(s_1,\ldots,s_n)$ {\it equivalent} if $s_i\in L_i$ for all $i$. Every ball embedding clearly admits an equivalent point embedding. Conversely, for any point embedding $\bm{s}=(s_1,\ldots,s_n)\in\mathbb{Z}_2^n$ with distinct coordinates, one can choose $k$ sufficiently large so that the balls
$L_i=s_i+2^k\mathbb{Z}_2$
are pairwise disjoint, yielding an equivalent ball embedding.

Several basic properties follow directly from the definitions:

\begin{theorem}\label{thm:embedBasic} The following properties of 2-adic embeddings hold:

    \begin{itemize}
        \item[1)] Equivalent point and ball embeddings induce the same tree: if $\mathcal{L}$ and $\bm{s}$ are equivalent, then $\mathcal{T}(\mathcal{L})=\mathcal{T}(\bm{s})$.
        
        \item[2)] If $\mathcal{L}$ and $\bm{s}$ are equivalent and $\mathcal{L}$ is full, then $\bm{s}$ is also full.

        \item[3)] Every full rooted binary tree $T$ admits a full ball embedding and, consequently, a full point embedding.
    \end{itemize}
\end{theorem}

\begin{proof}
In what follows, for $A\subseteq[n]$, let $S_A={s_i\in A}$ denote the corresponding set of embedded points.

    1) By the definition of a point embedding, a non-singleton set $A\subseteq[n]$ is a clade of $\mathcal{T}(\bm{s})$ if and only if
    $B(S_A)$ contains no $s_j$ with $j\notin A$. 
    Furthermore, laminarity of $2$-adic balls gives

    \begin{equation}\label{eq:equalBalls}
        B(S_A) = B\left(\bigcup_{i\in A} L_i\right) \text{ for any } A\subseteq [n], |A| \geq 2.
    \end{equation}

\noindent
Indeed, $B(S_A)$ intersects each $L_i$, $i\in A$. Because $|A|\geq2$ and the balls $L_i$ are pairwise disjoint, $B(S_A)$ cannot be contained in any $L_i$. By laminarity of $2$-adic balls, it therefore contains every $L_i$, and hence
$\bigcup_{i\in A}L_i\subseteq B(S_A)$. The reverse inclusion follows from the minimality of $B(S_A)$.

Therefore, for every non-singleton $A\subseteq[n]$, the ball on
the left of \eqref{eq:equalBalls} contains some $s_j$, $j\notin A$, if and only if the ball on the right intersects $L_j$. Thus $\mathbf{s}$ and $\mathcal{L}$ induce the same clades of the respective trees, and hence $\mathcal T(\mathbf s)=\mathcal T(\mathcal L)$.

2) Let $B$ be a ball such that $X = B\cap \{s_1,\ldots,s_n\}$ has at least two elements $s_i$ and $s_j$. Suppose that $X\subseteq B^{(0)}$ and $X\cap B^{(1)} = \emptyset$. Since the pairwise disjoint balls $L_i$ cover $\mathbb{Z}_2$, $B_1$ intersects  some $L_k$. Laminarity of 2-adic balls and the fact that $s_k\in L_k \setminus B^{(1)}$ thus imply that $B^{(1)}\subset L_k$. Now, since $B^{(1)}$ is a maximal sub-ball of $B$, we have $B \subseteq L_k$ and thus $L_k\ni s_i,s_j$, a contradiction.

3) Label the two outgoing edges of every internal vertex of $T$ by $0$ and $1$. For each vertex $v\in V(T)$ with the depth $d_v$, let $\bm{\alpha}^v = (\alpha^v_0,...,\alpha^v_{d_v-1})$ be the sequence of edge labels along the path from the root to $v$.  In particular, the words $\bm{\alpha}^i$ corresponding to the leaves form a complete prefix code \cite{cover1999elements}. Thus, by the Kraft equality \cite{cover1999elements} $\sum_{i=1}^n2^{-d_i}=1$.

Identify each vertex $v$ with the ball $L_v= \left(\sum_{r=0}^{d_v-1}\alpha^v_r 2^r\right) +2^{d_v}\mathbb Z_2$. The leaf balls $L_1,\ldots,L_n$ are pairwise disjoint and $\mu(L_i) = 2^{-d_i}$. The Kraft equality above imply that $\sum_{i=1}^n \mu(L_i) = 1$, and thus $L_1\sqcup\ldots \sqcup L_n = \mathbb{Z}_2$. So, the embedding $\mathcal{L} = \{L_1,...,L_n\}$ is full.

Furthermore, a leaf $i$ is a descendant of $v$ if and only if $\bm{\alpha}^v$ is a prefix of $\bm{\alpha}^i$, which is equivalent to $L_i\subseteq L_v$.  Therefore the
set of leaf balls contained in $L_v$ is exactly the clade of leaves descended from $v$.  Hence the clade family induced by
$L_1,\ldots,L_n$ coincides with the clade family of $T$. \end{proof}


Given the equivalence of the embedding variants, we will omit the embedding type when it is clear from context.

Any given tree $T$ admits many different embeddings. To associate $T$ with a single distinguished embedding, we define the {\it canonical point embedding} $\bm{s}=(s_1,\ldots,s_n)$ by the following conditions: 

\begin{itemize}
    \item[1)] $\bm{s}$ is full. 

    \item[2)] For any branching ball $B$ of $\bm{s}$ with left and right children $B^{(0)}$ and $B^{(1)}$, respectively, $\min\{i: s_i\in B^{(0)}\} < \min\{j: s_j\in B^{(1)}\}$.

    \item[3)]  For every $i\in [n]$, $0\leq s_i < 2^{1 + \max_{j\ne i}v_2(s_i - s_j)}$.
\end{itemize}

\begin{proposition}
    For every rooted labeled binary tree $T$, its canonical point embedding is unique.
\end{proposition}

\begin{proof}

Let $\bm{s}$ be a canonical embedding of $T$, and let $d_i$ be the depth of leaf $i$ in $T$.
The branching balls containing $s_i$ correspond exactly to the internal vertices on the root-to-$i$ path and hence there are $d_i$ of them.

Let $q_i=\max_{j\ne i}v_2(s_i-s_j)$, and $s_{i^*}$ is the coordinate where the maximum is achieved. Then the balls
$s_i+2^k\mathbb{Z}_2$, $k=0,\ldots,q_i$
contain at least two coordinates ($s_i$ and $s_{i^*}$). Since $\bm{s}$ is full, all of them are branching balls. The ball $s_i+2^{q_i+1}\mathbb{Z}_2$, on the other hand, contains no coordinate other than
$s_i$. Hence there are exactly $q_i+1$ branching balls containing $s_i$, and therefore $d_i=1+\max_{j\ne i}v_2(s_i-s_j)$.

Now consider the nested sequence $B_{i,0}\supset B_{i,1}\supset\cdots\supset B_{i,d_i-1}$
of successive branching balls containing $s_i$, and let $B_{i,d_i}$ be the child of $B_{i,d_i-1}$
containing $s_i$. By fullness, $B_{i,k+1}$ is a child of $B_{i,k}$ for every $k=0,\ldots,d_i-1$.

For each $k$, condition~2) uniquely determines whether $B_{i,k+1}$ is a left or right child of $B_{i,k}$. Indeed, if $C$ is a clade of $T$ corresponding to $B_{i,k}$, then $B_{i,k+1}$ is the left child if and only if $i$ belongs to the same maximal sub-clade of $C$ as the minimal element of $C$. Therefore the digits

\begin{equation}\label{eq:canon}
    a_{i,k}=
\begin{cases}
0,& B_{i,k+1}=B_{i,k}^{(0)},\\
1,& B_{i,k+1}=B_{i,k}^{(1)}
\end{cases}
\end{equation}

\noindent
are uniquely determined by $T$.

Consider $s_i' = \sum_{k=0}^{d_i-1}a_{i,k}2^k$. By the construction, $s_i' \equiv s_i \pmod{2^{d_i}}$ and $0 \leq s_i' \leq 2^{d_i}$. Therefore, the condition 3) implies that $s_i = s_i'$. Thus every coordinate $s_i$ is uniquely determined by $T$, and therefore so is the canonical embedding $\bm{s}$. The existence of such embedding follows from the fact that it can be constructed from a full embedding using the formula (\ref{eq:canon}). \end{proof}

\vspace{-2ex}
\section{Probabilistic properties of tree embeddings}\label{sec:measures}
\vspace{-1ex}

For a rooted binary tree $T$, let
$\mathcal{E}(T)=\{\bm{s}\in\mathbb{Z}_2^n:\mathcal{T}(\bm{s})=T\}$
denote the set of all point embeddings of $T$, and let
$\mathcal{E}_F(T)$
denotes the subset of $\mathcal{E}(T)$ consisting of full embeddings. Haar measures of these sets define the {\it Haar distribution} and {\it full Haar distribution} on rooted binary trees:
\begin{equation}
\chi(T):=\mu^n(\mathcal{E}(T)), \qquad
\chi_F(T):=\mu^n(\mathcal{E}_F(T)).
\end{equation}

\begin{theorem}\label{thm:measTreePoints}
    Let $T$ be a tree, and let $n_v = |C_{T,v}|$, $v\in I(T)$. Haar and full Haar tree distributions satisfy

    \begin{equation}\label{eq:probFixTree}
        \chi(T) = P(\mathcal{T}(\bm{s}) = T | s_1,...,s_n \stackrel{\text{i.i.d.}}{\sim} \mu) = \prod_{v\in I(T)}\frac{1}{2^{n_v-1}-1},
    \end{equation}

    \begin{equation}\label{eq:probFixFullTree}
        \chi_F(T) =  \prod_{v\in I(T)}\frac{1}{2^{n_v-1}},
    \end{equation}
\end{theorem}

\begin{proof}

Haar measure on $\mathbb{Z}_2$ is the pushforward of the Bernoulli product measure under the binary expansion map, i.e. every random variable $s_i$ can be represented as 

\begin{equation}\label{eq:2adBern}
        s_i = \sum_{k = 0}^{\infty} \alpha^i_k 2^k, \qquad \alpha^i_k \stackrel{\text{i.i.d.}}{\sim} \text{Bernoulli}(1/2).
\end{equation}

\noindent
Moreover, $\Pr(s_i=s_j)=0$ for $i\neq j$, so the embedding coordinates are distinct almost surely.

The random expansions (\ref{eq:2adBern}) generate the standard symmetric labelled binary trie \cite{fredkin1960trie}, with $\mathcal{T}(\bm{s})$ being its radix (PATRICIA) compression \cite{szpankowski1990patricia}. We then can use the formula for the probability of a plane unlabelled PATRICIA tree derived in \cite[Lemma~1]{ischebeck2026central}:

\begin{equation}
    \Pr(T_{\mathrm{plane,unlab}}) = n!\prod_{v\in I(T)}\frac{1}{2^{n_v}-2}.
\end{equation}

Fixing the leaf labels removes the factor $n!$, while ignoring the left--right ordering multiplies the probability by $2^{n-1}$, thus producing the formula (\ref{eq:probFixTree}). 

For a full embedding, at each node $v\in I(T)$ the prescribed unordered split of its $n_v$ descendant leaves into 2 smaller clades is realized by the digits
$(\alpha^i_{d_v})_{i\in C_v}\in\{0,1\}^{n_v}$, where $d_v$ is the depth of $v$. Exactly 2 of the $2^{n_v}$ possible digit assignments realize the particular split, so its probability is $2^{1-n_v}$. Conditional on this event, all higher-order digits of $s_i$ in the two daughter clades remain independent Bernoulli-\((1/2)\) variables. The same argument therefore applies recursively to the daughter clades, and multiplication over all internal
nodes gives \eqref{eq:probFixFullTree}.\end{proof}

The relations (\ref{eq:probFixTree}) - (\ref{eq:probFixFullTree}) establish connections between Haar tree distributions and {\it Sackin index} $S(T) = \sum_{i=1}^n d_i = \sum_{v\in I(T)} n_v$ -- a tree balance measure widely used in phylogenetics \cite{fischer2023tree}:

\begin{corollary}
    \begin{equation}\label{eq:sackin}
        S(T)-2(n-1) \leq -\log_2(\chi(T)) \leq S(T)-(n-1).
    \end{equation}
    \begin{equation}\label{eq:sackinfull}
        -\log_2(\chi_F(T)) = S(T) - (n-1).
    \end{equation}
\end{corollary}

By Theorem \ref{thm:measTreePoints}, the Haar tree distributions favor balanced trees. The classical Yule distribution \cite{yule1925ii} exhibits a similar preference, but its branching factors decay linearly with clade size, rather than exponentially.

In fact, every Aldous
$\beta$-splitting tree distribution with $\beta>-1$, including the Yule distribution, admits a natural
$2$-adic representation. For every dyadic ball $B\subset\mathbb Z_2$,
let $p_B\stackrel{\mathrm{i.i.d.}}{\sim}
\operatorname{Beta}(\beta+1,\beta+1)$.
Define a random probability measure $\nu_\beta$ on $\mathbb Z_2$
recursively by
\begin{equation}
\nu_\beta(B_0)=p_B\nu_\beta(B),\qquad
\nu_\beta(B_1)=(1-p_B)\nu_\beta(B),
\end{equation}
where $B_0$ and $B_1$ are the two child balls of $B$, and
$\nu_\beta(\mathbb Z_2)=1$.

Let $\mathcal Q_\beta$ denote the probability law of the random measure
$\nu_\beta$. Conditional on $\nu_\beta=\nu$, let
$s_1,\ldots,s_n$ be i.i.d. with distribution $\nu$, and define the
marginal distribution $\mathcal M_n^\beta$ of
$\bm s=(s_1,\ldots,s_n)$ by
\begin{equation}\label{eq:betaMarginal}
\mathcal M_n^\beta(A)
=
\int \nu^{\otimes n}(A)\,
\mathcal Q_\beta(d\nu),
\qquad
A\in\mathcal B(\mathbb Z_2^n).
\end{equation}

\begin{theorem}\label{thm:betasplit}
For every $\beta>-1$, the distribution of trees 
$\mathcal T(\bm s)$ for $\bm s\sim\mathcal M_n^\beta$
is the Aldous $\beta$-splitting distribution.
\end{theorem}

\begin{proof}
    Consider a clade $A\subseteq [n]$ of $\mathcal T(\bm s)$ with $|A|=m$. Let $B$ be a ball containing all points from $S_A$ and no points from $\{s_1,...,s_n\}\setminus S_A$, which exists by Theorem \ref{thm:embedBasic}, and let $B_0$ and $B_1$ be the child balls of $B$. Conditional on $p_B = p$, for any prescribed subset $X\subset A$ with $|X| = k$ we have

\begin{equation}
    \Pr(S_X\subseteq B_0,S_A\setminus S_X \subseteq B_1 \mid p_B = p) = p^k(1-p)^{m-k}.
\end{equation}

Since $p_B \sim
\operatorname{Beta}(\beta+1,\beta+1)$ with density  $f_{\beta+1}(p)=\frac{1}{B(\beta+1,\beta+1)} p^{\beta}(1-p)^{\beta}$, integration over $p_B$ gives

\begin{equation}
    \rho_{m,k} = \Pr(S_X\subseteq B_0,S_A\setminus S_X \subseteq B_1) = \int_0^1 p^k(1-p)^{m-k} f_{\beta+1}(p)\,dp = \frac{B(k+\beta+1,m-k+\beta + 1)}{B(\beta+1,\beta+1)}.
\end{equation}

Similarly, the probability of the trivial split at the ball $B$, with $X = \emptyset$ or $X = A$, is 

\begin{equation}
\tau_m = 2\frac{B(m+\beta+1,\beta+1)}{B(\beta+1,\beta+1)}.
\end{equation}

Now, the first nontrivial split of $A$ is the prescribed ordered split $(X,A\setminus X)$ whenever $S_X\subseteq B_0,S_A\setminus S_X \subseteq B_1$ for some ball $B$, with this event being preceded by $l$ trivial splits of $A$. Thus we have

\begin{equation}
    \Pr((X, A\setminus X)) = \sum_{l=0}^{\infty}\tau_m^l\rho_{m,k} = \frac{\rho_{m,k}}{1-\tau_m} = \frac{B(k+\beta+1,m-k+\beta + 1)}{B(\beta+1,\beta+1) - 2B(m+\beta+1,\beta+1)}
\end{equation}

Therefore the probability that the first nontrivial ordered split of the $m$-leaf clade has daughter-clade sizes $(k,m-k)$ is

\begin{equation}
    \Pr((k,m-k)) = \binom{m}{k} \frac{B(k+\beta+1,m-k+\beta + 1)}{B(\beta+1,\beta+1) - 2B(m+\beta+1,\beta+1)}.
\end{equation}

\noindent
This is exactly the splitting rule of the Aldous $\beta$-splitting model \cite{aldous2001stochastic}.\end{proof}

\vspace{-2ex}
\section{Algebraic properties of tree embeddings}
\vspace{-1ex}

\subsection{Algebraic operations on tree embeddings and their asymptotic behavior}

The ring structure of $\mathbb{Z}_2$ allows algebraic operations on point embeddings to be defined naturally. For example, we can define their sum and Hadamard product by
$\mathcal{T}(\bm{s})+\mathcal{T}(\bm{t}):=\mathcal{T}(\bm{s}+\bm{t})$
and
$\mathcal{T}(\bm{s})\odot\mathcal{T}(\bm{t}):=\mathcal{T}(\bm{s}\odot\bm{t})$,
provided that the resulting vectors have distinct coordinates.

In general, however, these operations depend on the particular embeddings rather than only on the underlying trees. For example, the embeddings
$\bm{s}=(0,2,1)$ and $\bm{t}=(0,2,7)$ induce the same tree,
$\mathcal{T}(\bm{s})=\mathcal{T}(\bm{t})=((1,2),3)$. Nevertheless,
$\mathcal{T}(\bm{s}+\bm{s})=((1,2),3)$, whereas $\mathcal{T}(\bm{s}+\bm{t})=((1,3),2)$. Thus, the tree induced by the sum is not determined by the trees induced by the summands.

The same holds for the Hadamard product. Let $\bm{s}=(1,2,3)$, $\bm{t}=(1,3,4)$, and $\bm{t}'=(2,4,5)$. Although $\mathcal{T}(\bm{t})=\mathcal{T}(\bm{t}')=((1,2),3)$, we have
$\mathcal{T}(\bm{s}\odot\bm{t})=(1,(2,3))$
and $\mathcal{T}(\bm{s}\odot\bm{t}')=((1,2),3)$. Therefore, to induce algebraic operations on trees themselves, we must first fix a canonical embedding for each tree and apply the operations to these canonical representatives.


We now analyze the limiting behavior of sums and Hadamard products of tree embeddings. Fix $n$ and let $\bm{s}^{(1)},\bm{s}^{(2)},\ldots$ be a sequence of independent random vectors in $\mathbb{Z}_2^n$ drawn from the same distribution $\mathcal{M}$. We assume that $\mathcal{M}$ is non-degenerate, meaning that $\mathcal{M}(A)>0$ for every non-empty open set $A\subseteq\mathbb{Z}_2^n$.

The following theorem shows that the tree induced by the sum of independent $2$-adic embeddings converges to the Haar tree distribution $\chi$, regardless of the initial distribution $\mathcal{M}$. Thus, in a certain sense, it provides a tree-valued analogue of the Central Limit Theorem.

\begin{theorem}\label{thm:HaarLimitSum}
Let $\bm s^{(1)},\bm s^{(2)},\ldots \stackrel{\mathrm{iid}}{\sim}\mathcal{M}$ and $\bm{S}^{(r)} = \sum_{i=1}^r \bm{s}^{(i)}$. For every rooted labeled binary tree $T$,

\begin{equation}\label{eq:HaarTreeLimit}
\lim_{r\to\infty} \Pr\left(\mathcal{T}(\bm{S^{(r)}}) = T\right) = \chi(T) = 
\prod_{v\in I(T)}\frac{1}{2^{n_v-1}-1}.
\end{equation}
\end{theorem}

\begin{proof}
 By the definition, the distribution of $\bm{S}^{(r)}$ is the convolution power $M^{*r}$. By the Kawada--It\^o theorem \cite{kawada1940probability,mustafayev2022note} applied to the additive group $\mathbb Z_2^n$, 

\begin{equation}\label{eq:KI}
M^{*r}\Longrightarrow\mu^n,
\end{equation}

\noindent
where $\Longrightarrow$ denotes weak convergence of probability measures.

It remains to pass from embeddings to tree topologies. Consider the tree $T$, with $\partial\mathcal{E}(T)$ being the boundary of its embedding set $\mathcal{E}(T)$.  Let

\begin{equation}
    D=\{\bm{t}\in\mathbb Z_2^n:t_i=j_j \text{ for some }i\ne j\}
\end{equation}

\noindent
be the collision set, $\bm{s}\in \mathcal{E}(T)$ be an embedding and $h>\max_{i\ne j} v_2(s_i-s_j)$ be a fixed integer. If $\bm{s}'\in \mathbb{Z}_2^n$ and    $\bm{s}'\equiv\bm{s}\pmod{2^h}$, then it is easy to see that $|s'_i - s'_j|_2 = |s_i - s_j|_2$ for all $i\ne j$, and hence
$\mathcal{T}(\bm s')= T$. Thus the embedding map $\bm{s} \mapsto \mathcal{T}(\bm{s})$ is locally constant on $\mathbb Z_2^n\setminus D$, and therefore $\partial\mathcal{E}(T)\subseteq D$.

Since independent random variables sampled from the Haar distribution coincide with probability zero, $\mu^n(D)=0$. Hence $\mathcal E(T)$ is a continuity set for $\mu^n$.
Equation (\ref{eq:KI}) and the Portmanteau theorem \cite{billingsley2013convergence} therefore imply that

\begin{equation}
    M^{*r}(\mathcal E(T)) \xrightarrow[r\to\infty]{} \mu^n(\mathcal E(T)).
\end{equation}

\noindent
Theorem~\ref{thm:measTreePoints} thus gives (\ref{eq:HaarTreeLimit}). 
\end{proof}

Interestingly, the asymptotic behavior of the Hadamard product is qualitatively different. Although the Haar distribution favors balanced trees, the following theorem shows that the product of independent Haar-distributed embeddings converges to the uniform distribution on caterpillar trees, which are maximally unbalanced. Recall that, for a permutation $\pi\in\bm{S}n$, the caterpillar $C{\pi}=(((\cdots((\pi_1,\pi_2),\pi_3),\pi_4),\ldots),\pi_n)$ is a rooted binary tree whose internal vertices form a single path. Here, $\pi_1$ and $\pi_2$ form the terminal cherry, while $\pi_3,\ldots,\pi_n$ are ordered successively toward the root.

\begin{theorem}\label{thm:HaarLimitProd}
For a fixed $n\geq 1$, let $\bm s^{(1)},\bm s^{(2)},\ldots \stackrel{\mathrm{iid}}{\sim}\mu^n$ and $\bm{Q}^{(r)} = \bm{s}^{(1)} \odot \ldots \odot \bm{s}^{(r)}$. Then, $\mathcal{T}(\bm{Q}^{(r)})$ converges in distribution to a caterpillar $C_{\pi}$, with $\pi \sim \mathrm{Unif}(\bm{S}_n)$ or, equivalently,   

\begin{equation}\label{eq:caterpillar}
    \lim_{r\to\infty}\Pr(\mathcal T(\bm{Q}^{(r)})=T) =
    \begin{cases}
    2/n!,&T\text{ is a caterpillar},\\
    0,&\text{otherwise}.
    \end{cases}    
\end{equation}
\end{theorem}

\begin{proof}
    For each coordinate of the product $\bm{Q}^{(r)}$ we have
    \begin{equation}
        v_2(\bm{Q}^{(r)}_i) = \sum_{\ell=1}^r v_2(s_i^{(\ell)}).
    \end{equation}

\noindent
As mentioned before (see the formula (\ref{eq:2adBern})), coefficients of the 2-adic expansions of each random point $s_i^{(\ell)}$ are independent symmetric Bernoulli variables. Recall that $v_2(s_i^{(\ell)}) = k$ if and only if $k$ is the minimal integer such that the corresponding expansion coefficient of $s_i^{(\ell)}$ is non-zero. Therefore, $v_2(s_i^{(\ell)})$ are independent geometric variables with mean $1$ and variance $2$. Hence

\begin{equation}
    \Pr\left(v_2(\bm{Q}^{(r)}_i)=k\right) = \binom{k+r-1}{r-1}2^{-k-r}, \qquad k\ge0.
\end{equation}

Let us now estimate the probability that two coordinates of $\bm{Q}^{(r)}$ have the same valuations. Let $(x)_m = x(x+1)\dots (x+m-1)$ be a Pochhammer symbol, ${}_2F_1(a, b; c; z)$ be a hypergeometric function and $P_m(x)$ be a Legendre polynomial. We have 

\begin{multline}
    p^{(r)}_{ij} = \Pr\left(v_2(\bm{Q}^{(r)}_i)=v_2(\bm{Q}^{(r)}_j)\right) = \sum_{k=0}^{\infty} \binom{k+r-1}{r-1}^22^{-2(k+r)} = 4^{-r} \sum_{k=0}^{\infty} \left(\frac{(r)_k} {k!}\right)^24^{-k} = \\ = 4^{-r} \sum_{k=0}^{\infty} \frac{(r)_k(r)_k}{(1)_k}\frac{4^{-k}}{k!} = 4^{-r}\cdot {}_2F_1(r, r; 1; 1/4).
\end{multline}

Now, using the Pfaff transformation ${}_2F_1(a, b; c; x) = (1-z)^{-b} {}_2F_1(c-a, b; c; \frac{x}{x-1})$ \cite{abramowitz1964handbook} and the identity ${}_2F_1(-m, m+1; 1; x) = P_m(1-2x)$ \cite{abramowitz1964handbook}[eq. 15.4.4], we get

\begin{equation}
    p^{(r)}_{ij} = 3^{-r} {}_2F_1(-(r-1), r; 1; -1/3) = 3^{-r}P_{r-1}(5/3).
\end{equation}

Laplace-Heine asymptotic formula for the Legendre polynomial \cite{szeg1939orthogonal}[eq. 8.21.1] thus gives

\begin{equation}
    p^{(r)}_{ij} \sim 3^{-r} \cdot \frac{1}{\sqrt{2\pi(r-1)}}\frac{(x + \sqrt{x^2-1})^{r-1/2}}{(x^2-1)^{1/4}}\Bigg|_{x=5/3} = 3^{-r}\cdot\frac{3^r}{\sqrt{8\pi r}} = \frac{1}{\sqrt{8\pi r}}.
\end{equation}

Consequently, 

\begin{equation}
    \Pr\left(v_2(Q^{(r)}_i)=v_2(Q^{(r)}_j) \text{ for some } i <j \right) \leq \binom{n}{2} p^{(r)}_{ij} = O\left(\frac{n^2}{\sqrt{r}}\right).
\end{equation}

\noindent
Thus, for the fixed $n$ and $r \rightarrow \infty$, coordinates of $\bm{Q}^{(r)}$ have pairwise distinct valuations with probability tending to one. It follows that $\mathcal{T}(\bm{Q}^{(r)})$ is a caterpillar. Indeed, let $k_i = v_2(Q^{(r)}_{i})$, and let us assume without loss of generality that $k_1 < k_2<\ldots < k_n$. For any $i < j$, the non-Archimedean property gives

\begin{equation}
    v_2\left(Q_{i}^{(r)}-Q_{j}^{(r)}\right) =
\min\{k_i,k_j\} = k_i.
\end{equation}

Hence, for every $i\in [n]$, the points $Q_{i}^{(r)},Q_{{i+1}}^{(r)},\ldots,Q_{n}^{(r)}$ belong to the ball $2^{k_{i}}\mathbb{Z}_2$, while all other points lie outside it. Therefore $A_i = \{i,i+1,...,n\}$ is a clade of $\mathcal{T}(\bm{Q}^{(r)})$, which is exactly a clade structure of a caterpillar.

Finally, note that conditional on distinct valuations, all $n!$ orders are equally likely. Each labeled caterpillar corresponds to two such orders, because the two leaves of the maximum depth can be exchanged. There are therefore $n!/2$ equally likely caterpillars, thus proving (\ref{eq:caterpillar}).
\end{proof}

\subsection{Polynomial and locally affine tree transformations}
Any function $\Phi:\mathbb{Z}_2^n\rightarrow\mathbb{Z}_2^n$ that maps valid embeddings to valid embeddings induces a transformation of embedded trees. In particular, any injective map $f:\mathbb{Z}_2\rightarrow\mathbb{Z}_2$ defines a transformation of tree embeddings,

\vspace{-1ex}
\begin{equation}
f(\mathcal{T}(\mathbf{s})) := 
\mathcal{T}((f(s_1),\ldots,f(s_n))).
\end{equation}

Of particular interest are functions whose induced tree transformations have a clear combinatorial interpretation. A simple example is $f(x)=x+1$, which swaps the two root subtrees of any embedded tree. More generally, standard rearrangements such as NNI, SPR, and TBR \cite{Felsenstein:Book:2003} can be realized by explicit algebraic transformations. 


\begin{figure}
    \centering
    \includegraphics[width=\textwidth]{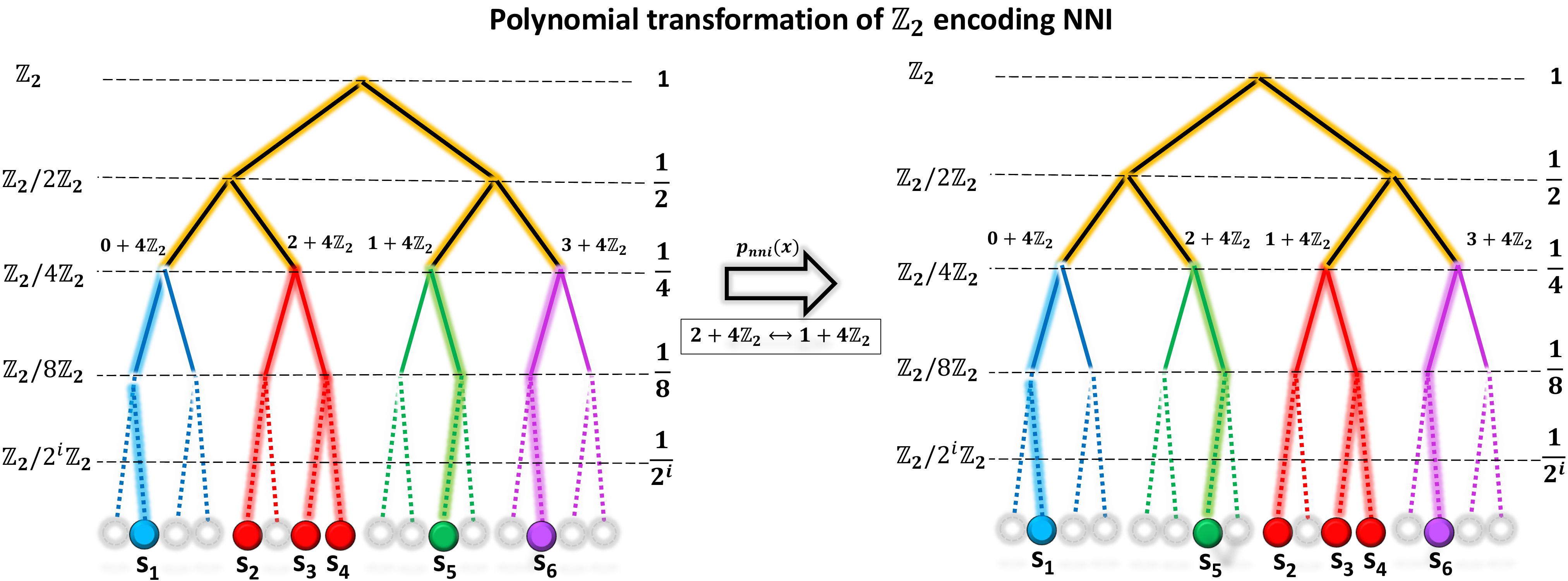}
    \caption{\small $2$-adic polynomial encoding of the NNI move \vspace{-2ex}}
    \label{fig:operationsNNI}
\end{figure}

\subsubsection{Basic tree rearrangements.}
As an example, consider the root NNI move that transforms the topology $T = \mathcal{T}(\bm{s}) = ((A,B),(C,D))$ into $T' = ((A,C),(B,D))$, where $A,B,C,D$ are clades. We can assume without loss of generality that $S_A \subset 4\mathbb Z_2$, $S_B \subset 2+4\mathbb Z_2$, $S_C \subset 1+4\mathbb Z_2$ and $S_D \subset 3+4\mathbb Z_2$. Let us call a function {\it locally affine}, if, for a set of disjoint balls $\{B_1,\ldots B_m\}$, $f|_{B_i} = a_i + b_ix$.

It is easy to see that the root NNI move is implemented by the locally affine function
\vspace{-2ex}
\begin{equation}\label{eq:NNIpieceaff}
f_{nni}(x) =
\begin{cases}
x, & x\equiv 0,3 \pmod{4},\\[4pt]
x-1, & x\equiv 2 \pmod{4},\\[4pt]
x+1, & x\equiv 1 \pmod{4}.
\end{cases}
\end{equation}
\noindent

Indeed, $f_{nni}(x)$ swaps the balls $1+4\mathbb Z_2$ and $2+4\mathbb Z_2$ maps the balls $4\mathbb Z_2$ and $3+4\mathbb Z_2$ into themselves, while acting as an isometry inside each of the residue classes modulo 4, i.e.

\vspace{-1ex}
\begin{equation}
|x-y|_2 = |f_{nni}(x)-f_{nni}(y)|_2
\qquad
\text{whenever } x \equiv y \pmod{4}.
\end{equation}

The local isometry guarantees that the relative $2$-adic distances within each residue class, and hence the topology of the corresponding subtrees of $\mathcal{T}(\bm{s})$, are preserved, while subtrees with clades $B$ and $C$ are swapped  (Fig. \ref{fig:operationsNNI}). Thus, $f_{nni}(\mathcal{T}(\mathbf{s})) = ((A,C),(B,D))$.

Instead of locally affine function, the same operation can be implemented by the single polynomial

\vspace{-1ex}
\begin{equation}\label{eq:NNI}
p_{nni}(x) = x + x^2 - \binom{x}{2}^4,
\end{equation}

\noindent
where $\binom{x}{k}=x(x-1)\cdots(x-k+1)/k!$ \footnote{Although the coefficients of $p_{nni}$ are rational, it is 2-adic integer-valued for $2$-adic integer arguments.}. 

\begin{proposition}
    $p_{nni}(x)$ acts like $f_{nni}(x)$ on residue classes modulo 4, and is an isometry inside each of them.
\end{proposition}

\begin{proof}
    It is easy to see that $x^2 \equiv 0 \pmod{4}$ for $x\equiv 0,2 \pmod{4}$ and $x^2 \equiv 1 \pmod{4}$ for $x\equiv 1,3 \pmod{4}$; in other words, it serves as an indicator function for odd classes modulo 4. Similarly, $\binom{x}{2}^4 \equiv 0 \pmod{4}$ for $x\equiv 0,1 \pmod{4}$ and $\binom{x}{2}^4 \equiv 1 \pmod{4}$ for $x\equiv 2,3 \pmod{4}$. So, $x^2 - \binom{x}{2}^4$ reproduce constant terms in (\ref{eq:NNIpieceaff}) modulo 4, and therefore $p_{nni}(x) \equiv f_{nni}(x) \pmod{4}$ for every $x \in \mathbb{Z}_2$.

To demonstrate isometry inside every mod-4 ball, note that 

\begin{equation}
    \frac{p_{nni}(x)-p_{nni}(y)}{x-y} = 1 + (x + y) - \frac{x+y-1}{2}\left(c(x) + c(y) \right)\left(c(x)^2 + c(y)^2 \right),
\end{equation}

\noindent
where $c(x) = \binom{x}{2}$. Suppose that $x - y \equiv 0\pmod{4}$. Then (a) $x+y$ is even; (b) $c(x) - c(y) = \frac{(x-y)(x+y-1)}{2}$ is even, i.e. $c(x)$ and $c(y)$ have the same parity. The latter implies that the term $\frac{x+y-1}{2}\left(c(x) + c(y) \right)\left(c(x)^2 + c(y)^2 \right)$ is even, and therefore $\frac{p_1(x)-p_1(y)}{x-y}$ is odd, thus giving $|p_{nni}(x) - p_{nni}(y)|_2 = |x-y|_2$. \end{proof}

\begin{figure}
    \centering
    \includegraphics[width=\textwidth]{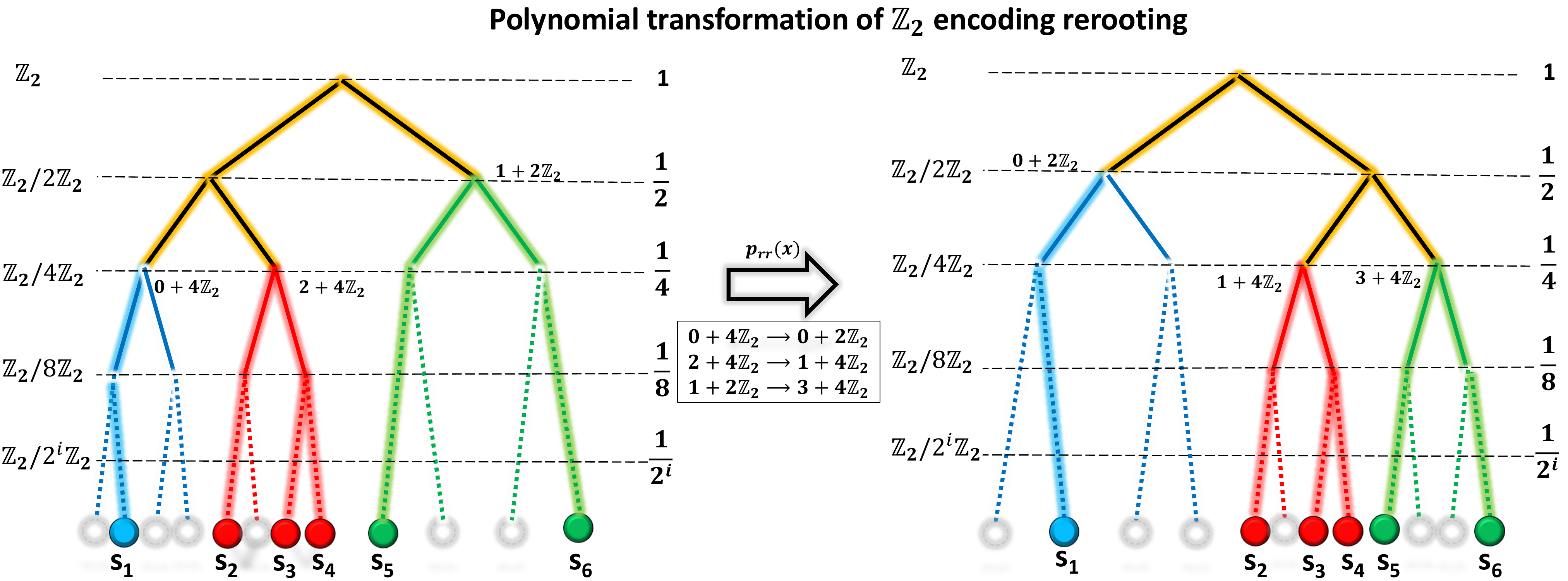}
    \caption{\small $2$-adic polynomial encoding of tree rerooting \vspace{-2ex}}
    \label{fig:operationsRR}
\end{figure}

Similarly, rerooting across an edge incident to the root, i.e. transforming $((A,B),C)$ into $(A,(B,C))$ (Fig. \ref{fig:operationsRR}), can be represented by the locally affine function 

\vspace{-2ex}
\begin{equation}
f_{rr}(x) =
\begin{cases}
\frac{1}{2}x, & x\equiv 0 \pmod{4},\\[4pt]
x-1, & x\equiv 2 \pmod{4},\\[4pt]
2x+1, & x\equiv 1 \pmod{2}.
\end{cases}
\end{equation}

\noindent
This function satisfies
$f_{rr}(4\mathbb Z_2)\subseteq2\mathbb Z_2$,
$f_{rr}(2+4\mathbb Z_2)\subseteq1+4\mathbb Z_2$, and
$f_{rr}(1+2\mathbb Z_2)\subseteq3+4\mathbb Z_2$, together with the local scaling property

\vspace{-2ex}
\begin{equation}
|f_{rr}(x)-f_{rr}(y)|_2 =
\begin{cases}
2|x-y|_2, & x,y\in4\mathbb Z_2,\\[4pt]
|x-y|_2, & x,y\in2+4\mathbb Z_2,\\[4pt]
\frac{1}{2}|x-y|_2, & x,y\in1+2\mathbb Z_2.
\end{cases}
\end{equation}

\noindent
The scaling guarantees that the relative distance relations between points inside the balls above are preserved, and thus topologies of the corresponding subtrees of $\mathcal{T}(\bm{s})$ are unaltered. The same effect can be achieved using the single 2-adic integer valued polynomial polynomial on integer arguments

\vspace{-2ex}
\begin{equation}\label{eq:rerootPol}
p_{rr}(x)=1+2x -\left(\frac{(x-2)(x^2-1)}2\right)^8.
\end{equation}

\begin{proposition}
    $p_{rr}(x)$ acts like $f_{rr}(x)$ on residue classes $4\mathbb{Z}_2$, $2+4\mathbb{Z}_2$ and $1+2\mathbb{Z}_2$, while being the similarity transformation with the scaling factor $\frac{1}{2}$ inside each of them.
\end{proposition}

\begin{proof}
    
    Let $q(x) = \frac{(x-2)(x^2-1)}{2}$ and thus $p_{rr}(x) = 1 + 2x - q(x)^8$. It is easy to see that $q(x) \equiv 1 \pmod{2}$ for $x\in 4\mathbb{Z}_2$, and $q(x) \equiv 0 \pmod{2}$ for $x\in (2+4\mathbb{Z}_2)\cup (1+2\mathbb{Z}_2)$. 

    Consequently, on $4\mathbb Z_2$, the two terms $1 + 2x+1$ and $s(x)^8$ are odd, so their difference is even. On the other two balls, $q(x)^8$ is divisible by $4$, so $p_{rr}(x)\equiv 1 + 2x \pmod{4}$. Thus $p_{rr}(4\mathbb Z_2)\subseteq 2\mathbb Z_2$,
    $p_{rr}(2+4\mathbb Z_2)\subseteq1+4\mathbb Z_2$, and
    $p_{rr}(1+2\mathbb Z_2)\subseteq3+4\mathbb Z_2$.
    
    To prove the scaling property, suppose that $x$ and $y$ belong to the same of the balls of interest. Since inside each of these balls $q(x)$ and $q(y)$ are of the same parity, each term $q(x)+q(y)$, $q(x)^2 + q(y)^2$ and $q(x)^4 + q(y)^4$ is even, and thus their product belongs to $8\mathbb{Z}_2$. Note also that $\frac{q(x)-q(y)}{x-y} = \frac{1}{2}(x^2 + xy + y^2 - 2x - 2y - 1)\in \frac{1}{2}\mathbb{Z}_2$. Together, these facts imply that 

    \begin{equation}
        \frac{q(x)^8 - q(y)^8}{x-y} = \frac{q(x)-q(y)}{x-y}(q(x)+q(y))(q(x)^2 + q(y)^2)(q(x)^4 + q(y)^4) \in 4\mathbb{Z}_2.
    \end{equation}

Hence
    \begin{equation}
    \frac{p_{rr}(x)-p_{rr}(y)}{x-y} = 2 - \frac{q(x)^8 - q(y)^8}{x-y} \in 2 + 4\mathbb{Z}_2.
\end{equation}

This gives the desired relation $|p_{rr}(x)-p_{rr}(y)|_2 = \frac{1}{2}|x-y|_2$. \end{proof}

The inverse rerooting, transforming $(A,(B,C))$ into $((A,B),C)$, is implemented by $p_{rr}(x+1)$. 


Note that a sequence of rerootings at root edge allow to move the root to an arbitrary tree edge; once the desired edge is at the root, the root NNI move performs an ordinary unrooted NNI on that edge. Since the NNI operation connects the space of unrooted binary phylogenies, the polynomials described above connect all rooted binary trees:

\begin{theorem}
    Let $P = \{p_1 = p_{nni},p_2 = p_{rr}, p_3 = x+1\}$ be a set of 3 polynomials. For every pair of binary trees $T_1$ and $T_2$ with embeddings $\mathcal{T}(\bm{s}_1)$ and $\mathcal{T}(\bm{s}_2)$ there exist a sequence $q_1,\ldots,q_k \in P$ such that $\mathcal{T}(\bm{s}_2) = (q_1\circ\ldots \circ q_k)(\mathcal{T}(\bm{s}_1))$.
\end{theorem}


\begin{figure}
    \centering
    \includegraphics[width=\textwidth]{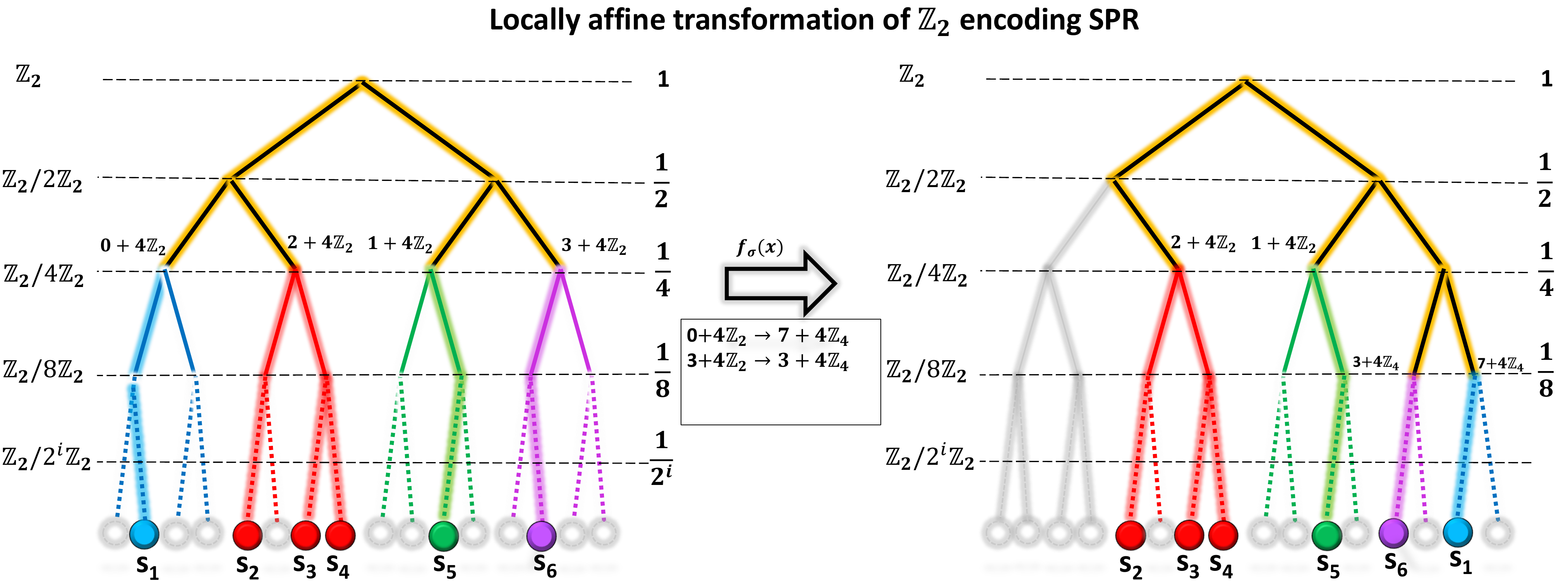}
    \caption{\small $2$-adic locally affine encoding of SPR \vspace{-2ex}}
    \label{fig:operationsSPR}
\end{figure}

SPR can be implemented in a similar manner  (Fig. \ref{fig:operationsSPR}). Suppose that $U$ is the clade of the tree $\mathcal{T}(\bm{s})$ to be pruned and $V$ is the clade below the target edge $e_V$ where $U$ to be re-attached. Denote the corresponding rearrangement operation as $\mathrm{SPR}_{U\rightarrow e_V}$. Let $B_U = B(S_U) = a_U +2^{d_U}\mathbb{Z}_2$ and $B_V = B(S_V) = a_V +2^{d_V}\mathbb{Z}_2$ be the minimal balls containing embedding points corresponding to the respective clades. Define the function $f_{\sigma}$ as follows:

\begin{equation}
    f_{\sigma}(x) =
\begin{cases}
2x-a_V, & x\in B_V,\\[2mm]
a_V+2^{d_V}+ 2^{d_V+1}\displaystyle\frac{x-a_U}{2^{d_U}}, & x\in B_U,\\[3mm]
x, & x\in \mathbb{Z}_2 \setminus (B_U\cup B_V),
\end{cases}
\end{equation}

\noindent
where $\sigma = (a_U,d_U,a_V,d_V)$. Note that the set $\mathbb{Z}_2 \setminus (B_U\cup B_V)$ is the disjoint union of balls.

This function maps $B_U$ and $B_V$ into 2 children balls of $B_U$, while, as in the previous examples, preserving topologies of the corresponding subtrees of $\mathcal{T}(\bm{s})$, as well as the relative locations of all leaves outside of $U\cup V$. Thus, indeed we have $\mathcal{T}(f_{\sigma}(\bm{s})) = \mathrm{SPR}_{U\rightarrow e_V}(\mathcal{T}(\bm{s}))$. Given that SPR connects the phylogeny space, this implies that the parametric family $(f_{\sigma})_{\sigma = (a_u,d_u,a_v,d_v)}$ connects rooted binary tree embeddings:

\begin{theorem}
    For every pair of binary trees $T_1$ and $T_2$ with embeddings $\mathcal{T}(\bm{s}_1)$ and $\mathcal{T}(\bm{s}_2)$ there exist a sequence $\sigma_1,\ldots,\sigma_k \in \mathbb{N}^4$ such that $\mathcal{T}(\bm{s}_2) = (f_{\sigma_1}\circ\ldots \circ f_{\sigma_k})(\mathcal{T}(\bm{s}_1))$.
\end{theorem}

\subsection{Locally scaling transformations and general rearrangements.}

In general, tree rearrangements like NNI and SPR, that modify the relations between selected subtrees of a given tree while preserving local tree structure, can be represented by {\it locally scaling transformations}. Formally, let $X,Y\subseteq \mathbb{Z}_2$ be compact-open sets, and let $\mathcal{B}=\{B_1,\ldots,B_L\}$ be a disjoint partition of  $X$ into balls which exists due to ultrametricity of $\mathbb{Z}_2$. We say that $f: X \rightarrow Y$ is {\it locally scaling}, if for every $B\in \mathcal{B}$, there exists $C_B > 0$ such
that 

\begin{equation}
    |f(x)-f(y)|_2=C_B|x-y|_2
\end{equation}

\noindent
for every $x,y \in B$.


Locally scaling transformations with $C_B \geq 1$ has been introduced and studied in \cite{kingsbery2009measure}. In our case, as the examples above show, scaling factors $C_B < 1$ are possible; however, several results of \cite{kingsbery2009measure} can be easily generalized to cases with arbitrary scaling factors. In particular, the following theorem generalizes all examples presented above.

\begin{theorem}\label{thm:locscale}
    Let $f:X\to Y$ be locally scaling with respect to a finite partition $\mathcal B$ of $X$ into $2$-adic balls, with scaling factors $(C_B)_{B\in\mathcal B}$. Then there exists a locally affine map $\ell:X\to Y$ and infinitely many polynomials $p\in\mathbb Q_2[x]$ such that, for every $B\in\mathcal B$,
    \begin{equation}
    f(B)=\ell(B)=p(B),
    \end{equation}
    and $\ell$ and $p$ have the same scaling factor $C_B$ as $f$ on $B$.
\end{theorem}

\begin{proof}
Theorem 7.2 of \cite{kingsbery2009measure} establishes the required locally affine and polynomial realizations when the map is a self-map, the underlying balls have a common radius, and all scaling factors satisfy $C_B\geq1$. We show that the general case considered here can be reduced to this setting by rescaling the domain and extending the resulting map to a locally scaling self-map of $\mathbb Z_2$.

Choose $k\ge0$ such that $2^kC_B\ge1$ for all $B\in\mathcal B$, set $\hat X=2^kX$, and define

\begin{equation}
\hat f(x)=f(2^{-k}x),\qquad x\in\hat{X}.
\end{equation}
\noindent
Then $\hat{f}$ has scaling factor $2^kC_B$ on $2^kB$.

Extend $\hat f$ to $\widetilde f:\mathbb Z_2\to\mathbb Z_2$ by setting

\begin{equation}
\widetilde f(x)=
\begin{cases}
\hat f(x),&x\in\hat X,\\
x,&x\notin\hat X.
\end{cases}
\end{equation}

\noindent
Since $\hat X=2^kX$ is a finite union of $2$-adic balls and, in an ultrametric space, the complement of such union is again a finite union of balls, both $\hat X$ and $\mathbb Z_2\setminus\hat X$ admit finite decompositions into $2$-adic balls. Thus, in an ultrametric space there exists a sufficiently small radius $r$ such that every ball in that decompositions is a union of $r$-balls. Thus $\widetilde f$ is locally scaling on the partition of $\mathbb Z_2$ into balls of radii $r$ ({\it $r$-balls}), with all scaling factors at least $1$.

By Theorem 7.2 and Lemma 7.1 of \cite{kingsbery2009measure}, there exist a locally affine map $\widetilde\ell:\mathbb Z_2\to\mathbb Z_2$ and infinitely many polynomials $\widetilde p\in\mathbb Q_2[x]$ such that, for every $r$-ball $\hat B$ of the chosen partition of $\mathbb Z_2$,

\begin{equation}
\widetilde\ell(\hat B) = \widetilde p(\hat B) = \widetilde f(\hat B),
\end{equation}
and $\widetilde\ell$ and $\widetilde p$
have the same scaling factor as $\widetilde f$
on $\hat B$. Since each ball $2^kB$, $B\in\mathcal B$,
is a union of balls from this $r$-ball partition,
it follows that

\begin{equation}
\widetilde\ell(2^kB) = \widetilde p(2^kB) = \widetilde f(2^kB) = f(B).
\end{equation}

Let us show that $\widetilde p$ and $\widetilde\ell$ have scaling factor $2^kC_B$ on the entire ball $2^kB$, not only on $r$-balls comprising it.
Consider $x,y\in 2^kB$. If $x$ and $y$ belong to the same $r$-ball
$\hat B$ of the chosen partition, then

\begin{equation}\label{eq:scaleExtendf}
|\widetilde p(x)-\widetilde p(y)|_2 = |\widetilde f(x)-\widetilde f(y)|_2
= 2^kC_B|x-y|_2.
\end{equation}

Suppose now that $x$ and $y$ belong to distinct
$r$-balls $\hat B_x$ and $\hat B_y$.
Then $|x-y|_2>r$, and therefore
\begin{equation}
|\widetilde f(x)-\widetilde f(y)|_2 = 2^kC_B|x-y|_2 > 2^kC_Br.
\end{equation}
\noindent
Thus $\widetilde f(\hat B_x)$ and $\widetilde f(\hat B_y)$ are balls of radius $2^kC_Br$ whose centers $\widetilde f(x)$ and $\widetilde f(y)$
are farther apart than their radius. Hence these balls are disjoint.

For two disjoint $2$-adic balls of the same radius, the distance between any point of one ball and any point of the other is constant and equals the distance between their centers. Therefore,  equation (\ref{eq:scaleExtendf}) holds, and thus the scaling of  $\widetilde p$ on the whole ball $2^kB$ is proved. The same argument applies to $\widetilde\ell$.

Finally, define

\begin{equation}
    \ell(x)=\widetilde{\ell}(2^kx),\qquad p(x)=\widetilde{p}(2^kx).
\end{equation}

\noindent
Then, for every $B\in \mathcal{B}$, $\ell(B)=p(B)=f(B)$ and for all $x,y\in B$, $|p(x)-p(y)|_2 =(2^kC_B)|2^k(x-y)|_2 =C_B|x-y|_2$, with the same identity for $\ell$. Since distinct polynomials $\widetilde p$ yield distinct polynomials $p$, infinitely many polynomial realizations are obtained. \end{proof}


\subsection{Exploring tree space via 2-adic polynomial local search}
\vspace{-1ex}

Theorem \ref{thm:locscale} demonstrates that the class of tree transformations defined by locally affine and polynomial functions is substantially richer than the limited family of standard tree moves.  This suggests a new type of local search or MCMC sampling in which moves are selected not from a fixed repertoire of combinatorial tree rearrangements, but from broad families of $2$-adic polynomials.


Formally, consider a family of functions $\mathcal{F}=\{f_i: \mathbb{Z}_2 \rightarrow \mathbb{Z}_2 : i \in I\}$. For the given embedding $T = \mathcal{T}(\bm{s})$, $\mathbf{s}=(s_1,\ldots,s_n)$, its neighborhood is defined as 

\vspace{-1ex}
\begin{equation}
\mathcal{N}(T) = \{f_i(T) = 
\mathcal{T}(f_i(s_1),\ldots,f_i(s_n)): i\in I\}.
\end{equation}

\noindent
At each step, the objective is evaluated on the transformed trees and the best candidate is selected, or candidates are sampled according to
a stochastic rule such as simulated annealing or MCMC. Since only finite trees
are considered, all computations can be performed at finite $2$-adic precision. Importantly, the transformation family $\mathcal{F}$ can be made problem-aware, rather than being restricted
to a small set of generic moves as in conventional tree search. It may include both classical rearrangements and problem-specific transformations, with the potential to accelerate convergence, improve solution quality, and escape local optima that trap standard methods.

Consider a parametric family of polynomials 

\begin{equation}\label{eq:pencil}
    \mathcal{P}_{g,h} = \{p_{\theta}(x) = g(x) - \theta h(x) : \theta \in \mathbb{Z}_2\},
\end{equation}

\noindent
where $\theta \in \mathbb{Z}_2$ is a parameter, and $g,h\in \mathbb{Q}[x]$ are 2-adic integer-valued polynomials.  Local search over the family $g-\theta h$ has a natural analogy with directional optimization: one chooses a direction $h$, optimizes over the step parameter $\theta$, and moves along the resulting direction. Conceptually, this is closely related to multidirectional line-search methods in conventional numerical optimization.

Let $\bm{s}\in \mathbb{Z}_2^n$ be a fixed embedding. Then $p_{\theta}(\bm{s})$ is a proper embedding, i.e. $p_{\theta}(s_i) \ne p_{\theta}(s_j)$ if and only if $\theta$ belongs to the {\it admissible set}

\begin{equation}
   \Xi = \left\{q_{ij} = \frac{g(s_i) - g(s_j)}{h(s_i) - h(s_j)}: 1 \leq i < j \leq n, h(s_i) \ne h(s_j)\right\},\quad \Theta = \mathbb{Z}_2 \setminus \Xi.
\end{equation}

\noindent
Suppose also that $(g(s_i),h(s_i))\ne (g(s_j),h(s_j))$ for $i\ne j$. The following theorem shows that {\it varying the coefficient $\theta$ generates only at most quadratically many distinct trees}, independently of the degrees of the polynomials $g$ and $h$.

\begin{theorem}\label{thm:localSearchRegions}
    For every fixed embedding $\bm{s}\in \mathbb{Z}_2^n$ there exists a partition $\mathcal{R} = \mathcal{R}_{\bm{s}}$ of $\mathbb{Z}_2$ into at most $(n-1)(n-2)+1$ clopen sets such that $p_{\theta}(\mathcal{T}(\bm{s}))$ is constant as a function of admissible $\theta$ on each $R\in \mathcal{R}$, i.e.

    \begin{equation}\label{eq:treeEquiv}
        p_{\theta}(\mathcal{T}(\bm{s})) = p_{\theta'}(\mathcal{T}(\bm{s})) \text{ for any } \theta,\theta'\in R\cap \Theta, R \in \mathcal{R}.
    \end{equation}
    \noindent
    Moreover, $R\cap \Theta \ne \emptyset$ for all $R \in \mathcal{R}$.
\end{theorem}

\begin{proof} 
    In what follows, we consider $\theta\in \Theta$. Let $q_{ij} = \frac{g(s_i) - g(s_j)}{h(s_i) - h(s_j)}$. Let us also assume for now that $h|_{\bm{s}}$ is injective, i.e. $h(s_i) \ne h(s_j)$ for $i\ne j$ (later this constraint will be lifted). 
    
    It is easy to check that the following identities hold:
    
    \begin{equation}\label{eq:valptheya}
        p_\theta(s_i)-p_\theta(s_j) =(h(s_i)-h(s_j))(q_{ij}-\theta),
    \end{equation}
    
    \begin{equation} \label{eq:ratioIdentity} 
        (h(s_i)-h(s_k))(q_{ik}-q_{jk}) = (h(s_i)-h(s_j))(q_{ij}-q_{jk}),
    \end{equation}
    
    The first of these identities imply that 

    \begin{equation}\label{eq:pairwiseDist}
        |p_\theta(s_i)-p_\theta(s_j)|_2 = |h(s_i)-h(s_j)|_2|q_{ij}-\theta|_2.
    \end{equation}

    \noindent
    or, in the valuation form,

    \begin{equation}\label{eq:pairwiseValuation}
        v_2\bigl(p_\theta(s_i)-p_\theta(s_j)\bigr) = v_2(h(s_i)-h(s_j))+v_2(q_{ij}-\theta).
    \end{equation}
    
    Denote by $\operatorname{lca}_T(i,j)$ the lowest common ancestor of leaves $i$ and $j$ in the rooted binary tree $T$. Let us call three distinct leaves $i,j,k$ of $T$ {\it a rooted triple} and write $(ij,k)$, if $\operatorname{lca}_T(i,j)$ is a strict descendant of $\operatorname{lca}_T(i,k) = \operatorname{lca}_T(j,k)$ or, equivalently, $i$ and $j$ are closer to each other than to $k$. 

     Consider the trees $P_{\theta} = p_{\theta}(\mathcal{T}(s_1,...,s_n))$, $H = h(\mathcal{T}(s_1,...,s_n))$ and a rooted triple $(ij,k)$ in $H$. Then

     \begin{equation}\label{eq:rootTriple}
         \alpha := v_2(h(s_i) - h(s_j)) > \beta:= v_2(h(s_i) - h(s_k)) = v_2(h(s_j) - h(s_k)).
     \end{equation}

    Now, our goal is to trace what happens with relative location of leaves $i,j,k$ in the tree $P_{\theta}$. This is defined by the change in pairwise distances between these leaves in $H$ and $P_{\theta}$ which is, in turn, described by the equality (\ref{eq:pairwiseDist}).
     
     Assume that $q_{ik} = q_{jk}$. In this case, (\ref{eq:ratioIdentity}) implies that all 3 coefficients $q_{ij}$, $q_{ik}$ and $q_{jk}$ are equal. Therefore, by (\ref{eq:pairwiseDist}), pairwise distances between any pair of leaves from $\{i,j,k\}$ in the trees $P_{\theta}$ and $H$ differ by the same multiplier, and so these leaves form the same rooted triple $(ij,k)$ in $P_{\theta}$.

     Let now $q_{ik} \ne q_{jk}$, and thus ${\varepsilon} = v_2(q_{ik} - q_{jk})$ is finite, i.e. $q_{ik} - q_{jk} = 2^{{\varepsilon}}u$, $u \in \mathbb{Z}_2^{\times}$. Define two balls in $\mathbb{Q}_2$:
    \begin{equation}
    \label{eq:critical-balls-proof}
    U_{(ik,j)}=q_{ik}+2^{{\varepsilon}+1}\mathbb Z_2,
    \qquad
    U_{(jk,i)}=q_{jk}+2^{{\varepsilon}+1}\mathbb Z_2.
    \end{equation}

    \noindent
    Note that they are sibling balls with the common parent ball $U_{ijk} = q_{ik}+2^{{\varepsilon}}\mathbb Z_2 = q_{jk}+2^{{\varepsilon}}\mathbb Z_2$. This follows from the fact that $q_{ik}$ and $q_{jk}$ both belong to $U_{ijk}$ and are not congruent modulo $2^{\varepsilon + 1}$; therefore they belong to different children balls of $U_{ijk}$, thus also being their centers.

     \begin{lemma}
         For every $\theta\in \Theta$,

        \begin{equation}\label{eq:tripleTransition}
        P_{\theta}|_{\{i,j,k\}}
        =
        \begin{cases}
        (ik,j),&\theta\in U_{(ik,j)},\\
        (jk,i),&\theta\in U_{(jk,i)},\\
        (ij,k),&\theta\notin U_{(ik,j)}\sqcup U_{(jk,i)} = U_{ijk}.
        \end{cases}
        \end{equation}
     \end{lemma}
     \begin{proof}
    
    Taking valuations in \eqref{eq:ratioIdentity} and using (\ref{eq:rootTriple}), we get

    \begin{equation}\label{eq:delta}
        \delta := v_2(q_{ij}-q_{jk}) = {\varepsilon} - ( \alpha - \beta) < {\varepsilon}
    \end{equation}

    \noindent
    Similarly, $v_2(q_{ij}-q_{ik}) < {\varepsilon}$.

    Now assume that $\theta\in U_{(ik,j)}$, i.e. $q_{ik} - \theta \equiv 0 \pmod{2^{{\varepsilon}+1}}$. Let us assess the valuations of $q_{ik} - \theta$, $q_{jk} - \theta$ and $q_{ij} - \theta$.  
    For the first difference, we immediately have $v_2(q_{ik} - \theta) \geq {\varepsilon}+1$. For the second difference, $q_{jk} - \theta \equiv q_{jk} - q_{ik} \equiv 2^{{\varepsilon}}u \pmod{2^{{\varepsilon}+1}} $, which gives $v_2(q_{jk} - \theta) = {\varepsilon}$. Finally, for the third difference, the ultrametric equality gives

    \begin{equation}\label{eq:qij}
        v_2(q_{ij} - \theta) = v_2((q_{ij} - q_{jk}) + (q_{jk} - \theta)) = \min\{\delta,{\varepsilon}\} = \delta.
    \end{equation}

    Now, by (\ref{eq:pairwiseValuation}), we have

    \begin{equation}
        v_2\bigl(p_\theta(s_i)-p_\theta(s_k)\bigr) \geq \beta + \varepsilon + 1,\quad  v_2\bigl(p_\theta(s_j)-p_\theta(s_k)\bigr) = \beta + \varepsilon, \quad  v_2\bigl(p_\theta(s_i)-p_\theta(s_j)\bigr) = \alpha +\delta =  \beta + \varepsilon.
    \end{equation}

    \noindent
    So, in the tree $P$ leaves $i$ and $k$ are closer to each other than to $j$, and thus form the triple $(ik,j)$. The case $\theta \in U_{(jk,i)}$ is symmetric.

    It remains to consider the case $\theta\notin U_{ijk}$. In this case, $\lambda := v_2(q_{ik} - \theta) = v_2(q_{jk} - \theta) < \varepsilon$, and thus, from (\ref{eq:pairwiseValuation}),

    \begin{equation}
        v_2\bigl(p_\theta(s_i)-p_\theta(s_k)\bigr) = v_2\bigl(p_\theta(s_j)-p_\theta(s_k)\bigr) = \beta + \lambda.
    \end{equation}
    
    Similarly to (\ref{eq:qij}), the relation (\ref{eq:delta}) and the ultrametric inequality give $v_2(q_{ij} - \theta) \geq \min\{\delta,\lambda\}$. Thus we have 

     \begin{equation}
        v_2\bigl(p_\theta(s_i)-p_\theta(s_j)\bigr)
        \geq \alpha + \min\{\delta,\lambda\} =
        \begin{cases}
        \alpha + \lambda > \beta  + \lambda, \;\; \lambda \leq \delta\\
        \alpha + \delta = \beta + \varepsilon > \beta + \lambda.
        \end{cases}
        \end{equation}
    So, $i$ and $j$ are closer to each other than to $k$ in $P$, and thus form the triple $(ij,k)$.
     \end{proof}

    For every rooted triple $(ij,k)$ of $H$, let $C_{(ij,k)} = U_{(ij,k)}\cap \mathbb{Z}_2$; since $\mathbb{Z}_2 = B(0,1)$ is a ball, so is $C_{(ij,k)}$. We call these balls the \emph{critical balls} associated with the triple $(ij,k)$. Note that critical balls associated with different triples can coincide, and some balls may be trivial, i.e. $C_{(ij,k)} \in \{\emptyset,\mathbb{Z}_2\}$. 

    \begin{lemma}
        Let $\mathcal{C}$ be the set of all non-trivial critical balls. Then $|\mathcal{C}| \leq (n-1)(n-2)$.
    \end{lemma}
    \begin{proof}
        For an internal vertex $v\in V(H)$, let $A_v$ and $B_v$ be the clades corresponding to its two child subtrees. Define
        \begin{equation}\label{eq:cross-pair-ratios}
        \mathcal{Q}_v = \{q_{ij}\in \Xi:i\in A_v,\ j\in B_v\}.
        \end{equation}

        \noindent
        Obviously, it is true that $|\mathcal{Q}_v|\leq |A_v||B_v|$.
        
        For a rooted triple $(ij,k)$, let $v$ be the lowest common ancestor of its three leaves in $H$. Without loss of generality, $i,j\in A_v$ and $k\in B_v$. Hence $q_{ik},q_{jk}\in\mathcal{Q}_v$. If these ratios are distinct, the parent ball $C_{ijk}$ of their two critical balls is a branching ball of $\mathcal{Q}_v$. Thus, for every $v\in I(H)$, the number of parent balls of critical balls of rooted triples separated by $v$ does not exceed the number of  branching balls of $\mathcal{Q}_v$. Recall also that for the set $\mathcal{Q}_v$ there are exactly $|\mathcal{Q}_v|-1$ branching balls.   It follows that
        
        \begin{equation}
        \label{eq:critical-ball-count}
        |\mathcal{C}|
        \le
        2\sum_{v\in I(H)}\bigl(|\mathcal{Q}_v|-1\bigr)
        \le
        2\sum_{v\in I(H)}\bigl(|A_v||B_v|-1\bigr) = 2\sum_{v\in I(H)}|A_v||B_v| - 2(n-1).
        \end{equation}

        Every unordered leaf pair lies in different child subtrees at
        exactly one vertex, its lowest common ancestor. Therefore $\sum_{v\in I(H)}|A_v||B_v|=\binom{n}{2}$, and hence
        
        \begin{equation}
        \label{eq:critical-ball-final-count}
        |\mathcal C| \leq 2\binom{n}{2} - 2(n-1) =(n-1)(n-2).
        \end{equation}
    \end{proof}
    
    Let us partition the parameter space $\mathbb{Z}_2$ by declaring two parameters $\theta_1$ and $\theta_2$ equivalent whenever they belong to exactly the same members of $\mathcal{C}$. Let $\mathcal{R}$ be the resulting set of equivalence classes.

    If $\theta,\theta' \in \Theta$ are equivalent, then, according to (\ref{eq:tripleTransition}), the trees $P_{\theta_1}$ and $P_{\theta_2}$ have the same rooted triples for every three leaves $i,j,k\in [n]$. Since a labeled rooted binary tree is determined by its rooted triples, (\ref{eq:treeEquiv}) follows.

    To count the number of equivalence classes, we will use the laminarity of the ball set $\mathcal{C}$. Consider the extended set $\hat{\mathcal{C}} = \mathcal{C} \cup \{\mathbb{Z}_2\}$, which is also laminar So, $\hat{\mathcal{C}}$ is an inclusion poset whose Hasse diagram $\mathcal{H}$ is a tree. For $B\in \hat{\mathcal{C}}$ let

    \begin{equation}\label{eq:RBdef}
        R_B = B\setminus\bigcup_{C\in \mathrm{ch}_{\mathcal{H}}(B)}C,
    \end{equation}

    \noindent
    where $\mathrm{ch}_{\mathcal{H}}(B)$ be the children of $B$ in $\mathcal{H}$, i.e. the maximal critical balls properly contained in $B$. Thus, $\theta \in R_B$ if and only if $B$ is the minimal critical ball containing $\theta$ (with $B = \mathbb{Z}_2$ whenever $\theta$ does not belong to any critical ball). By laminarity, all parameters in $R_B$ belong to exactly the same critical balls, namely $B$, if \(B\in\mathcal C\), and all critical balls containing $B$. Therefore $\mathcal{R} = \{R_B \ne \emptyset : B\in \hat{\mathcal{C}}\}$ is the desired partition of the parameter space into equivalence classes, and 

    \begin{equation}
        |\mathcal{R}| \leq |\mathcal{C}|+1 \leq (n-1)(n-2)+1.
    \end{equation}

    Furthermore, $R_B = B\bigcap (\bigcap_{C\in\mathrm{ch}_{\mathcal H}(B)} (\mathbb Z_2\setminus C))$; thus, it can be obtain by a finite number of interections and complements of clopen 2-adic balls, which makes it clopen as well. 

    All the considerations above were made under the assumption that $h(s_i) \ne h(s_j)$ for $i\ne j$. We will now show that the established results hold without this assumption. Suppose only that $g(s_i) \ne g(s_j)$ when $h(s_i) = h(s_j)$. Define the function

    \begin{equation}
        \tilde{h}(x) = h(x) + \varepsilon g(x),
    \end{equation}
    where $\varepsilon \in 2\mathbb{Z}_2$ is is chosen so that the values $\tilde{h}(s_i)$, $i\in [n]$, are distinct. Such $\varepsilon$ exists, because the number of values of $\varepsilon$ where $\tilde{h}|_{\bm{s}}$ is non-injective is finite.

    Now set $\phi = \frac{\theta}{1+\varepsilon \theta}$. Since $\varepsilon$ is even, the denominator here is odd and therefore invertible in $\mathbb{Z}_2$, giving $\phi\in \mathbb{Z}_2$. The preceding construction can therefore be applied to $\tilde{p}_{\phi}(x) = g(x)-\phi_{\theta}\widetilde h(x)$. Moreover,  it is easy to check that 

    \begin{equation}\label{eq:paramChange}
        p_{\theta}(x) = g(x)-\theta h(x) =  (1+\varepsilon\theta)\tilde{p}_{\phi}(x).
    \end{equation}

    \noindent
    The function $p_{\theta}$ and $\tilde{p}_{\phi}$ differ by an odd multiplicative constant; thus, distances between images of any pair of point are the same for these functions, and consequently $p_\theta$ and $\widetilde p_\phi$ produce the same tree on $\bm{s}$.
    Similarly, the mapping $\theta\mapsto\phi = \phi(\theta)$ is a bijective isometry of $\mathbb Z_2$. Hence critical balls in the
    $\phi$-parameter space correspond bijectively to critical balls
    of the same radii in the $\theta$-parameter space, and the same
    is true for the associated parameter regions. Therefore all the
    preceding results extend to the present case. This completes the proof. \end{proof}

    \begin{corollary}
        For a fixed embedding $\bm{s}\in \mathbb{Z}_2^n$ and any polynomial $f \in \mathbb{Z}_2[x]$, varying any single
        coefficient of $f$ while keeping all other coefficients fixed produces at most $(n-1)(n-2)+1$ distinct valid trees.
\end{corollary}

Theorem \ref{thm:localSearchRegions} demonstrates that, although an embedding of a given tree is not unique and the parametric family of polynomials used for local search can be infinite, it is possible to organize a search in such a way that the consecutive search iterations produce different trees. 

Its proof also suggests an algorithm for generating this neighborhood. Indeed, every parent ball of a pair of sibling critical balls is a branching ball of a subset of $\Xi$, and hence is also a branching ball of $\Xi$. Consequently, using the children of all branching balls of $\Xi$ may produce a refinement of the parameter-space subdivision defined by the critical balls, but $\mathcal{T}(p_\theta(\bm{s}))$ is still constant within every region of this refinement.

In what follows, for $q\in\mathbb Z_2$ let $\varphi_k(q)=((q \bmod 2^{k+1})+2^k)\bmod 2^{k+1}$. In words, $\varphi_k(q)$ is the number whose 2-adic expansion retains the first $k$ digits of $q$, flips the next digit, and sets all subsequent digits to zero. Thus,
\begin{equation}
v_2(\varphi_k(q)-q)=k,
\end{equation}

\noindent
so $\varphi_k(q)$ belongs to the ball $q+2^k\mathbb Z_2$ but not to $q+2^{k+1}\mathbb Z_2$. We use $\varphi_k$ to construct explicit representatives of parameter regions: by choosing $k$ appropriately, $\varphi_k(q)$ is placed inside a prescribed larger ball while remaining outside a specified smaller ball containing $q$.

The neighborhood-generation procedure is described in Algorithm \ref{alg:cap}. To estimate its running time, let $m=|\Xi|\le\binom n2$, $d_g=\deg g$, and $d_h=\deg h$. Furthermore, let 

\begin{equation}
    D_{\Xi} = 1+\max_{q\ne q'\in\Xi}v_2(q-q'),\qquad D_p = \max_{\theta\in\Theta_{\rm rep}}
\left(1+\max_{i\ne j}v_2\bigl(p_\theta(s_i)-p_\theta(s_j)\bigr)\right)
\end{equation}

\noindent
be the numbers of 2-adic digits needed to distinguish the elements of $\Xi$ and coordinates $p_{\theta}(s_i)$, respectively. 

It can be shown that $D_{\Xi} = O(D_p)$. Indeed, let $q = q_{ij}$ and $q'$ be collision parameter maximizers for $D_{\Xi}$, and $u$ be their parent. Then $\varepsilon_u=v_2(q-q')=D_{\Xi}-1$. The representative parameter generated by $q$ is $\theta = \varphi_{\varepsilon_u+1}(q)$  (l. 22 of Algorithm \ref{alg:cap}) and therefore, by the definition of $\varphi$, $v_2(q_{ij} - \theta) = \varepsilon_u + 1 = D_{\Xi}$. Then using the identity (\ref{eq:pairwiseValuation}), 

\begin{equation}
    D_p \geq 1 + v_2(p_{\theta}(s_i) - p_{\theta}(s_j)) = v_2(b_i-b_j) + v_2(q_{ij} - \theta) \geq D_{\Xi} + 1
\end{equation}

Now consider the particular stages of Algorithm \ref{alg:cap}. Computing $a_i=g(s_i)$ and $b_i=h(s_i)$ by Horner's rule takes $O(n(d_g+d_h))$ arithmetic operations, and constructing the $m$ collision parameters requires $O(n^2)$ operations. The radix tree of $\Xi$ can be constructed in $O(mD_{\Xi}) = O(mD_p)$ bit operations and has $O(m)$ vertices; thus the algorithm generates at most $2m-1$ representative parameters. For each representative parameter $\theta$, computing the transformed coordinates $p_\theta(s_i)=a_i-\theta b_i$ takes $O(n)$ arithmetic operations, while constructing their radix tree takes $O(nD_p)$ bit operations. Hence the total running time is

\begin{equation}
O\bigl(n(d_g+d_h)+n^2 + mnD_p\bigr).
\end{equation}

In practical settings, $d_g,d_h=O(n)$, and in many phylogenetic models, the trees of interest have logarithmic height. If, in addition, the embeddings $\mathbf{s}$ are full, i.e. we canonicalize the embedding at each step of the local search, then $D_p=O(\log n)$. In this case the total running time becomes
\begin{equation}
O\bigl(n^2+mn\log n\bigr).
\end{equation}


\begin{algorithm}[H]
\setlength{\baselineskip}{0.9\baselineskip}
\caption{Generation of parameter-region representatives and the corresponding tree neighborhood}
\label{alg:cap}
\begin{algorithmic}[1]
    \Require Polynomials $g,h\in\mathbb Z_2[x]$ and an embedding
    vector $\bm{s}\in\mathbb Z_2^n$, with
    $T=\mathcal{T}(\bm{s})$.
    
    \Ensure A set of representative parameters
    $\Theta_{\mathrm{rep}}$ and the corresponding set of distinct
    tree topologies $\mathcal{N}(T)$.

    \State $\Theta_{\mathrm{rep}}\gets\emptyset$.

    \State Compute
    $a_i=g(s_i)$ and $b_i=h(s_i)$ for $i=1,\ldots,n$.

    \If{$(a_i,b_i)=(a_j,b_j)$ for some $i\ne j$}
        \State \Return $\Theta_{\mathrm{rep}},\mathcal N(T)$
    \EndIf

    \State Construct
    $\Xi\gets \left\{ \frac{a_i-a_j}{b_i-b_j}:
    1\le i<j\le n,\; b_i\ne b_j\right\}\cap \mathbb{Z}_2$; $m\gets|\Xi|$
    \Comment{collision parameters}

    \If{$m=0$}
        \State \Return $\Theta_{\mathrm{rep}}=\{0\}$
        \Comment{the topology is constant on $\mathbb Z_2$}
    \EndIf

    \If{$m=1$, i.e. $\Xi=\{q\}$}
        \State $\Theta_{\mathrm{rep}}=\{\varphi_0(q)\}$
        \Comment{a single representative from $\mathbb Z_2\setminus\{q\}$}
    \EndIf

    \If{$m\geq 2$}
    \State Construct the compressed binary 2-adic radix tree
    $R$ of $\Xi$, with $L(R)=\Xi$.

    \State For every vertex $v\in V(R)$ store any $q_v\in C_v$; if $v\in I(R)$, also store
    $\varepsilon_v\gets\min\{v_2(q-q'):q,q'\in C_v,\;q\ne q'\}$. Equivalently, $q_v+2^{\varepsilon_v}\mathbb Z_2$ is the minimal ball containing $C_v$

    \For{$v$ in a preorder traversal of $R$}

        \If{$v=r(R)$ and $\varepsilon_v>0$}
            \State $\theta\gets\varphi_{\varepsilon_v-1}(q_v)$
            \Comment{representative outside the minimal ball containing $\Xi$}

        \ElsIf{$v\in I(R)\setminus\{r(R)\}$ and
        $\varepsilon_v>\varepsilon_{\mathrm{parent}(v)}+1$}
            \State $\theta\gets\varphi_{\varepsilon_v-1}(q_v)$
            \Comment{representative between the parent child ball and the ball containing $C_v$}

        \ElsIf{$v\in L(R)$}
            \State $\theta\gets
            \varphi_{\varepsilon_{\mathrm{parent}(v)}+1}(q_v)$
            \Comment{representative of the terminal ball distinct from $q_v$}

        \EndIf

        \If{a parameter $\theta$ is generated}
            \State $\Theta_{\mathrm{rep}}
            \gets\Theta_{\mathrm{rep}}\cup\{\theta\}$.
        \EndIf
    \EndFor
    \EndIf
    \State Calculate trees $\mathcal{N}(T) = \{\mathcal{T}(a_1 - \theta b_1,\ldots,a_n - \theta b_n): \theta \in \Theta_{\mathrm{rep}}\}$
    \State \Return $\Theta_{\mathrm{rep}},\mathcal N(T)$

\end{algorithmic}
\end{algorithm}

\section{Discussion}

Cross-fertilization between theoretical and applied fields has historically played a crucial role in scientific progress. Computational biology, and phylogenetics in particular, has a rich tradition of applying algebraic methods \cite{casanellas2019algebraic,pachter2005algebraic,pachter2004tropical,sullivant2023algebraic,bocci2007topics,eriksson2005phylogenetic,rhodes2012identifiability,HendySnir-TCBB-2008,HendyPenSte-PNAS-1994,hendy1989framework,baez2015operads,monod2018tropical,sturmfels2005toric,hollering2021identifiability,diaconis1998matchings}. These methods are mathematically powerful and have yielded deep insights into the structure, identifiability, and geometry of phylogenetic models. Their impact on practical phylogenetic inference, however, has remained comparatively limited.

Our broader goal is to explore how methods from discrete optimization, phylogenetics, and non-Archimedean analysis can be combined to develop an {\it optimization framework that is both theoretically sound and practically effective}. This paper is a first step toward this goal. We show that the $2$-adic embeddings studied here provide a useful mathematical language and a methodological basis for formulating several concepts that may contribute to a broader theory of tree calculus and optimization, including algebraic operations on trees, tree measures, and tree limits. Our approach also allows phylogenetic inference to be viewed, in a certain sense, as an embedding problem: given biological sequences $g_1,\ldots,g_n$, the goal is to assign them coordinates in a numerical space so as to optimize a specified objective. This creates a direct analogy with other sequence-embedding problems, with the Euclidean space replaced by the non-Archimedean space of 2-adic integers.

One specific outcome of this study is a generalized local-search approach based on polynomial transformations. This approach offers several potential advantages. First, the richness of the resulting neighborhood structures, together with efficient methods for exploring them, allows search moves to be adapted to particular phylogenetic inference problems. Second, the numerical nature of the search enables efficient arithmetic or bit-level implementations and facilitates reuse of computations within a common ambient $2$-adic space: coordinate-wise transformation values, residue-class tests, and related quantities can be precomputed or cached and reused across multiple tree states. By contrast, standard methods often operate primarily on complete individual tree topologies, making exhaustive precomputation over tree space infeasible.

In the future, the $2$-adic approach may support the development of new optimization methods, potentially including gradient-based or convex techniques, with the goal of improving convergence and solution quality in phylogenetic inference. More broadly, it may provide a basis for algorithmic approaches drawing on algebraic geometry, ergodic theory, and non-Archimedean analysis. Such methods could enable not only the design of new algorithms, but also the study of their convergence, stability, and mixing behavior by viewing them as dynamical systems on $p$-adic spaces. Beyond phylogenetics, the framework may also extend to optimization problems involving trees and hierarchical structures in other application domains. Realizing these possibilities will require substantial further work.

\bibliographystyle{unsrt}
\bibliography{references}

\end{document}